\documentclass{article}
\usepackage{fullpage}

\usepackage{soul}
\usepackage{url}
\usepackage[utf8]{inputenc}
\usepackage[small]{caption}
\usepackage{graphicx}
\usepackage{amsmath}
\usepackage{amsthm}
\usepackage{booktabs}
\usepackage{algorithm}
\usepackage{algorithmic}
\usepackage[switch]{lineno}

\usepackage{xgames}

\usepackage[round]{natbib}
\usepackage[colorlinks=true,linkcolor=blue,citecolor=blue]{hyperref}

\newtheorem{example}{Example}
\newtheorem{theorem}{Theorem}

\usepackage[utf8]{inputenc}
\usepackage{amsmath,amssymb}

\usepackage{xspace}

\usepackage{amsmath,graphicx}
\usepackage{float}
\usepackage[subrefformat=parens,labelformat=parens]{subcaption}

\usepackage{enumitem}
\usepackage{tabularx}

\newcommand{\EE}[1]{\displaystyle \mathop{\mathbb{E}}_{#1}}

\usepackage[nameinlink]{cleveref}

\usepackage{color}
\definecolor{emphcol}{gray}{.75}

\newcommand{\x}{\mathbf{x}}

\newcommand{\argmax}{\textrm{argmax}}

\newcommand{\high}[1]{#1_{\mathsf{high}}}
\newcommand{\low}[1]{#1_{\mathsf{low}}}

\newcommand{\profit}{\mathsf{profit}}

\newcommand{\potmax}{\mathsf{POTMAX}}

\newcommand{\appsec}{Appendix}

\renewcommand{\cite}[1]{\citep{#1}}
\newtheorem{definition}{Definition}
\newtheorem{corollary}{Corollary}
\newtheorem{lemma}{Lemma}
\newtheorem{remark}{Remark}

\begin{document}
	
	\title{\bf 
		Airdrop  Games
	}


	\author{
		Sotiris Georganas \\
		IOG \\
		\texttt{sotiris.georganas@iohk.io}
		\and 
		  Aggelos Kiayias\\
		  University of Edinburgh, IOG\\
		  \texttt{akiayias@inf.ed.ac.uk}
		  \and
		  Paolo Penna\\
		  IOG\\
		  \texttt{paolo.penna@iohk.io}
		}

		
	\date{\today}

	\maketitle

	\begin{abstract}

Launching a new blockchain system or application is frequently facilitated 
by a so called {\em airdrop}, where the system designer 
chooses a  pre-existing set of potentially interested parties and  allocates
newly minted tokens to them with the expectation 
that they will participate in  the system --- such engagement, especially if it is of  significant level, facilitates the system and raises its value and also 
the value of its newly minted token,
hence benefiting the airdrop recipients. 
A number of challenging questions befuddle designers in this setting, such as how to choose the set of interested parties and how to allocate tokens to them. 
To address these considerations we put forward a game-theoretic model
for such {\em airdrop games}.
Our model can be used to guide
the designer's choices based on the way the system's value
depends on participation (modeled by a ``technology function'' in our framework) and the costs that participants incur. We identify both bad and good equilibria 
and identify the settings and the choices that can be made where the designer can influence the players 
towards good equilibria in an expedient manner. 

\end{abstract}

\section{Introduction}
Launching a new blockchain system is challenging as it requires the upfront contributions of different parties, without any guarantee that the system will be successful. The characteristics of such launches are as follows: 
\begin{itemize}
    \item There is a set of possibly interested parties.
    Participating incurs some \emph{cost}, hence the system designer performs an \emph{airdrop} of tokens to entice the participants: a certain amount of the available tokens are distributed in  advance to potential contributors, regardless of their (future) individual contribution to the system \cite{why}.\footnote{ The nature of  participation or contribution should be interpreted broadly and includes holding tokens, participating in governance, or actively running bespoke software that performs system functions.} Identifying the potential contributors typically piggybacks on an existing blockchain system e.g., as in ``restaking'' in Ethereum where new tokens are allocated based on existing staked ether holdings \cite{eigenlayer}, but more direct approaches have also been attempted, e.g., in worldcoin \cite{worldcoin},
    prospective users scan their retina in order to receive tokens.
    \item The eventual success of the system depends on the actual contributions and level of participation, which, in turn, reflects on the monetary value of the tokens received via the airdrop. The higher the overall participation of the players, the higher is the value of the new token, and thus also the value of the airdrop allocation received initially.  The dependency between system value and participation can be modeled by an underlying {\em technology function} that we make explicit below. 
\end{itemize}

Potential contributors thus face a dilemma: If they contribute, they incur a cost but (potentially) increase the value of their token allocation. Naturally, contributors should act strategically and contribute in a way that  maximizes utility. Several equilibria exist: in good equilibria, ``enough participation'' is achieved and the launch of the system ``succeeds'' while in the bad equilibria a complete breakdown of the system is possible. 

From the designer's perspective, some fundamental questions need to be addressed in order to understand how a project can be successfully launched:
\begin{quote}
    \emph{What is the level of contribution of the parties that we can reasonably expect, given a specific allocation?}
    \emph{How can this be influenced by different ``tokenomics'' policies that award larger or smaller amounts of tokens as part of the airdrop allocation?}
    \emph{What kind of technology functions are more favorable in terms of facilitating a successful launch?}
\end{quote} 
In this work, we formally address these questions via a novel game-theoretic model and its analysis. To illustrate the nature of the problems, consider the following technology function:
\begin{example}[Threshold Technologies.]\label{ex:threshold-intro} Consider a system technology that requires contributions from at least 50\% of the  contributors (the total number being $N=10$). If this \emph{threshold} is met, the system operates correctly, and the token's value is \emph{high}, \$10. Conversely, if the threshold is not reached, the system fails, and the token's value drops to \emph{low}, at 0. With an airdrop granting each participant 1 token and contribution cost $\alpha = 1$, two equilibria emerge: (i) no one contributes, since an individual contribution alone does not increase the token's value but incurs a cost of \$1, and (ii) exactly 50\% contribute, that is 5 players. If any one of them contemplated not contributing, that  would cause the token's value to drop from \$10 to 0 (a net profit change from 10 - 1= 9 to 0). Also none of the 5 players not contributing has any incentive to contribute, since they are already enjoying the high value without any cost. Clearly, the latter equilibrium is preferable, and its existence for higher costs is guaranteed only if the designer sets the airdrop properly (e.g., for participation costs of $\alpha = 20$ USD, an individual airdrop of at least 2 tokens would be required).
\end{example}

The threshold-based technology described above is natural, but we can consider \emph{other types} of technologies determining the system's value. For instance, \cite{ALABI201723} suggests that several systems follow \emph{Metcalfe's Law}: the value of a network is proportional to the square of the number $\ell$ of contributors (number of entities \emph{holding} the native token in a wallet). In this case, the token's value follows $t(\ell) = q \cdot \ell^2$, with $q$ a rescaling constant. The equilibria emerging with this technology are \emph{different} from the example above, leading to different trade-offs when deciding on the number of tokens to airdrop.

\subsection{Our Contribution}
Our contributions can be summarized as follows. 
\paragraph{Game-theoretic model (Section~\ref{sec:model}).}
We propose a  game-theoretic model for airdrops. The model incorporates the key feature that token allocations are issued in some \emph{new} token whose value: (i) is not determined at launch but, (ii) is affected by the actual participation or contribution of the potentially interested parties allocated these tokens according to some
{\em technology function}. 
The analysis of equilibria in our model (see below) is informative to the designer and answers basic questions like ``\emph{Who receives tokens/how many?}'' as posed in \cite{frowis2019operational}. Intuitively, the designer sets some ``eligibility'' criterion based on past information, which determines the number and the individual costs (``who receives tokens and why''), and  the corresponding airdrop allocation (``how many tokens'').

    \paragraph{Analysis  of equilibria (Section~\ref{sec:airdrop-general}).} We characterize the set of pure Nash equilibria for the general setting, as a function of \emph{amount of rewarded tokens}, the \emph{number of potential contributors}, their \emph{individual costs}, and the \emph{technology} which ``converts'' individual contributions into system value. We show that the model's general version corresponds to a \emph{potential game} (Theorem~\ref{th:potential-game}). Thus,  pure Nash equilibria always exist and are reached via  simple best response dynamics. We further characterize the set of pure Nash equilibria  in Section~\ref{sec:aridrop-pne}. It is worth noting that (i) we consider \emph{heterogeneous} costs, that is, players have different costs in general,  (ii) pure Nash equilibria do \emph{not} require players to know about others' costs but only about others' strategies (contributions), and (iii) these equilibria are quite natural as they arise from simple best response (as opposed to \emph{mixed} Nash equilibria for which no ``simple'' dynamics exist, and whose empirical support is comparably limited at the level of individual player behavior in the lab or field). These considerations are also fundamental for the designer, as a target.  
    
    \paragraph{Refined unique  (logit)  equilibria (Section~\ref{sec:airdop-logit}).} For some technologies, \emph{bad} equilibria where no player contributes coexist with \emph{good} equilibria where a sufficiently high level of players contribution is reached, thus making the system valuable. We consider a well-known class of ``noisy'' best response dynamics, termed \emph{logit dynamics} \citep{blume1993statistical,BLUME2003251} (see Section~\ref{sec:related-work} for further discussion) for our model. These seem natural in our context and do not require excessive sophistication from the players. We first show that, in the so-called \emph{vanishing noise regime}, these dynamics select only  \emph{stochastically stable} pure Nash equilibria, characterized by Theorem~\ref{th:stable-equilibria}, thus  discarding \emph{bad} equilibria in several cases (see below).  We also consider the so-called \emph{finite noise} regime and the corresponding \emph{unique} stationary equilibrium \cite{auletta2011convergence}. Under some mild restrictions on the model (Section~\ref{sec:binary}), we provide tight bounds on the \emph{time} for the dynamics to reach its equilibrium or a particular level of contribution (Section~\ref{sec:logit-time}).  Time is also a fundamental aspect for the designer, as the system needs to reach a ``sufficiently good'' state within the given launch period (after which the system is supposed to start its normal autonomous operations).

    \paragraph{Applications to relevant technology functions (Section~\ref{sec:threshold}).} Our model accommodates generic technology functions to express the token value depending on the actual contributions of the tokens' receivers. We apply our results to the important class of  \emph{threshold} technologies, a natural, simple  description of systems which need an ``initial minimal base'' to succeed  \citep{chaidos2023blockchain,crwodfundinggame,yan2021optimal,math9212757,chang2020economics,10.1109/TNET.2023.3274114,wang2023bundling}. Threshold technologies are a typical example where \emph{bad} equilibria with zero participation always exist, along with \emph{good} equilibria, so there is a need for theory to explain how successful systems can reach good equilibria. Threshold technologies represent a ``hard case'' in our setting, in the sense that bad equilibria persist for any airdrop reward amount. Logit dynamics provide a  formal argument that good equilibria are more likely to be chosen (Theorem~\ref{th:pot-max-threshold} and Corollary~\ref{cor:stable-theshold}) also in this hard case. 
    The analysis  further indicates when it is optimal for the designer to perform an airdrop at all (and the optimal rewards) in both the vanishing- and non-vanishing regimes (Sections~\ref{sec:threshold-vanishing} and \ref{sec:threshold-nonvanishing}).

       The analysis of threshold technologies highlights the need for \emph{low costs} to converge to a desired state within reasonable time (Section~\ref{sec:threshold-time}).   
       This highlights the benefits of recent technological developments like Ethereum's restaking \cite{eigenlayer} or Cardano's \emph{partnerchain} framework \citep{IOHK-PC-02}, which ``reuse'' contributors from a mainchain who are likely to incur lower costs (\appsec~\ref{sec:partnerchain}).

We stress that our results can be applied to other technology functions, including Metcalfe's Law mentioned above, and other examples in the literature (see \appsec~\ref{sec:other-technoligies}).

\subsection{Related Work}\label{sec:related-work}
\paragraph{Airdrops.} Airdrops are costly for the designer, as recording them on an existing chain incurs transaction fees, necessitating \emph{simple} allocation strategies \cite{frowis2019operational,designing}. Empirical studies suggest simple airdrops will remain common for “projects without established on-chain activity” \cite{evolutionary}, while others \cite{harder} recommend ``scaling rewards with costs'', aligning with our theoretical findings. Additional work on practical features and goals of airdrops includes \cite{rise,tierdrop,altruistic,designing}.

\paragraph{Related Games and Models.}
Our model can be seen as a variant of \emph{blockchain participation games} \citep{chaidos2023blockchain} and the \emph{combinatorial agency} model \citep{BABAIOFF2012999} in contract theory. In the former, players receive a \emph{monetary} reward, contrary to our setting where the rewards are tied to  the system value  (the variant of universal payments with no retraction is the closest to ours, whereas in other variants rewards are even more loosely tied to the system's success);  the system value is a threshold technology of (eligible) players actively contributing. In the combinatorial agency model, players receive again a \emph{monetary} reward conditioned on the success of the project, expressed by some ``success probability'' function on the contributing players. Consequently, the equilibria in these two models are different from ours and these results do not apply to our setting.

    A closely related class of problems is \emph{crowdfunding} games \cite{crwodfundinggame,yan2021optimal,math9212757,chang2020economics,10.1109/TNET.2023.3274114,wang2023bundling}. Similar studies concern public funding projects \citep{soundy2021game,DBLP:journals/tcs/BiloGM23} and public goods projects on networks \citep{bramoulle2007public,galeotti2010network,dall2011optimal,yu2020computing}. These models are mathematically equivalent to our games in the case where the designer \emph{cannot change the rewards}. Specific problems correspond to technology functions being S-shaped \citep{buragohain2003game} or a threshold function \citep{galeotti2010network}. Stochastic stability is studied in \cite{boncinelli2012stochastic}. All these problems (and ours) belong to rich  class of public goods games. These are generally computationally hard, even when the underlying (technology) function is specified by a network   \cite{10.1145/3465456.3467616,gilboa2022complexity,10.1145/3580507.3597780,do2024tight,10.1257/aer.100.4.1468}.

\paragraph{Logit Dynamics.}\label{eq:logit-dynamics-connections}
Logit dynamics have been largely studied in the context of games and equilibrium selection problem, that is, as a refinement of pure Nash equilibrium  (see e.g. \cite{blume1993statistical,BLUME2003251,10.1109/FOCS.2009.64,asadpour2009inefficiency,alos2010logit,alos2015robust,auletta2011convergence,auletta2012metastability,auletta2013mixing,auletta2013logit,okada2012log,coucheney2014general,ferraioli2016decentralized,MamageishviliP16,ferraioli2017social,alos2017convergence,DBLP:journals/ijgt/Penna18,kleer2021sampling}). This is usually done in two ways. The first is to consider the so-called \emph{vanishing noise} regimes  and a resulting set of \emph{stable equilibria}  (see e.g. \cite{blume1993statistical,alos2010logit,alos2015robust}). The second is to consider \emph{non-vanishing noise} regimes and the corresponding unique stationary distribution of the process as the equilibrium concept \cite{auletta2011convergence}.
The logit response model also finds applications in economics \citep{costain2019logit}, in  pricing algorithms \citep{muller2021dynamic,van2022price}, and coalitional bargaining \citep{sawa2019stochastic}.

\section{Modeling Airdrop Games}\label{sec:model}
The model captures the following key aspects of airdrops:
(i) The designer chooses the amount of tokens to be airdropped to potential contributors. 
   (ii) Potential contributors decide whether to perform a (costly) task for the system. The resulting system value depends on the total contribution and the underlying ``technology'' of the project. 
    (iii) Contributors maximize utility, resulting in an \emph{equilbrium} and system value.

The designer faces a tradeoff between the airdrop 
amount and the resulting system value (too small airdrops do not incentivize enough contributors and thus the project fails, while giving away all tokens is not optimal either because it minimizes profit).

\newcommand{\R}{\mathbb{R}}
\newcommand{\vv}[1]{\mathbf{#1}}
\subsection{Parameters and Underlying Game}

    \paragraph{Contributors.} There is a set of $n$  potential contributors (players): Each contributor chooses her  actual contribution (strategy)  $a_i\in A_i \subseteq \R^+$, incurring in a cost of   $c_i \cdot a_i$ where $c_i$ denotes the per unit cost of $i$.  
	\paragraph{System (technology) value.} The overall value of the system depends on each individual effort or contribution, i.e., on the strategy profile $a= (a_1,\ldots,a_n)$ and it is equal to $V(a)$ for some monotone non-decreasing function (higher contributions yield the same or higher value). 
	\paragraph{Token value (and total supply).} Given the token total supply (the overall number of tokens of in system) $T_{tot}$, the value or price of the token is
	\begin{align}\label{eq:token-val}
	t(a) := \frac{	V(a)}{T_{tot}} \ . 
	\end{align}
	\paragraph{Airdrop (Token) Allocation.} The designer allocates some constant fraction $\rho\in[0,1]$ of the overall token supply as an airdrop, i.e.,  to be distributed equally among the players. Thus, each player receives $\gamma$ tokens, where
	\begin{align}\label{eq:rewards}
		\gamma := \frac{\rho\cdot T_{tot}}{n} && \rho \in [0,1] \ . 
	\end{align}
 (For the sake of simplicity, we allow $\gamma$ to be a fractional number, and ignore rounding effects.)
It is worth noting  that the monetary reward (number of tokens times the token price) is \emph{independent} on the token total supply $T_{tot}$, and it equals to a fraction $\rho/n$ of the system value,
\begin{align}\label{eq:inv-monetary-rewards}
	\gamma \cdot t(a)  \stackrel{\eqref{eq:token-val}+\eqref{eq:rewards}}{=} \frac{\rho\cdot T_{tot}}{n}\cdot \frac{V(a)}{T_{tot}} = \frac{\rho}{n} \cdot V(a) \ . 
\end{align}
	\paragraph{Utilities.} Players' utilities equals the monetary reward received (number of tokens times the token value) minus the incurred cost
	\begin{align}\label{eq:utiliy}
		u_i(a) := \gamma \cdot t(a) - c_i \cdot a_i \stackrel{\eqref{eq:inv-monetary-rewards}}{=} \frac{\rho}{n}\cdot V(a) - c_i \cdot a_i  \ . 
	\end{align}
    \paragraph{Equilibria.} A strategy profile $a=(a_1,\ldots,a_n)$ is a pure Nash equilibrium if no player $i$ can increase her utility by changing her strategy $a_i$, that is, 
    \begin{align}\label{eq:equilibrium-def}
    	u_i(a) \geq  u_i(a') && \stackrel{\eqref{eq:utiliy}}{\Longleftrightarrow} &&  \frac{\rho}{n}\cdot (V(a) - V(a')) \geq  c_i \cdot (a_i -a_i')
    \end{align}
for all $i$ and all $a' = (a_1,\ldots, a_{i-1},a_i',a_{i+1},\ldots,a_n)$. 

\paragraph{Logit dynamics.}
Logit dynamics  \citep{blume1993statistical,BLUME2003251} are a kind of   ``noisy'' best response  dynamics where players have some \emph{inverse noise} or \emph{learning rate}
$\beta\geq 0$ and each of them responds according to a so-called \emph{logit response} 
\begin{align}\label{eq:logit-response}
    p_i^{\beta}(a_i|a_{-i}) = \frac{\exp(\beta u_i(a_i,a_{-i}))}{Z_i^\beta}
\end{align}
where $Z_i^{\beta}$ is a normalizing factor so that the above formula is a probability distribution, and $(x,a_{-i}):=(a_1,\ldots,a_{-i},x,a_{i+1},\ldots,a_n)$. For $\beta=0$, players choose a strategy at random with uniform distribution, while for $\beta\rightarrow \infty$, they tend to the best response rule.\footnote{In case multiple best response exist, the corresponding player chooses any of them with uniform distribution.} At each step of the dynamics, a randomly chosen player revises her current strategy according to the logit response above \eqref{eq:logit-response}.  
Logit dynamics converge to a \emph{unique equilibrium} $\pi^{\beta}$ given by the  \emph{stationary distribution}  of the underlying Markov chain:  $\pi^{\beta}(a)$ is the probability of players selecting profile $a$  after sufficiently many steps of revisions (``learning'') have been performed. Note that this depends on the parameter $\beta$.  For the class of exact potential games, in the vanishing noise regime ($\beta\rightarrow \infty$), equilibrium $\pi^{\beta}$ concentrates on the \emph{subset} of pure Nash equilibria whose potential is optimal \citep{blume1993statistical,BLUME2003251}.  The  \emph{mixing time} of the underlying Markov chain is the time required by the dynamics to reach the equilibrium $\pi^{\beta}$ starting from \emph{any} state \citep{LevinPeresWilmer2006,auletta2011convergence}.

\subsection{Objectives and Metrics} There are different (possibly conflicting) metrics to evaluate system performance, given contributions,  costs, the designer's profit etc. 
\paragraph{System Value.} This is the value specified by the technology function $V(a)$ as a function of all contributions. 
    \paragraph{Social Cost.} This is the sum of all players costs,
    \begin{align}\label{eq:SC}
        SC(a,c) := \sum_i c_i \cdot a_i \ . 
    \end{align}
    \paragraph{Users' Welfare.} This is the sum of all players' utilities \eqref{eq:utiliy},\begin{align}
        \label{eq:SW}
        UW(a,\rho,c):= \sum_i u_i(a) \stackrel{\eqref{eq:utiliy}+\eqref{eq:SC}}{=} \rho\cdot V(a) - SC(a,c) &\ . 
    \end{align}
    
    \paragraph{Profit.} This is the value of the remaining tokens remaining with the designer, after subtracting the airdropped tokens and the cost $d_V$ for developing the technology $V(\cdot)$, 
    \begin{align}\label{eq:profit}
       \profit(a,\rho) = (1-\rho) \cdot V(a)  - d_V\ .
    \end{align}
    The designer strategically chooses the airdrop allocation $\rho$ aiming to maximize its profit (utility) defined as in \eqref{eq:profit}. 
    \footnote{Note that we still assume the designer to move first by announcing the airdrop $\rho$, and then the (other) players will reach some  equilibrium accordingly.}
We consider which system values can be achieved given the underlying technology, players' costs, and their utility maximizing strategies. Note that the designer's profit can be \emph{negative} (also when the former does not provide any reward, $\rho=0$). This corresponds to inherently ``bad'' projects that are  destined to fail and should not be implemented. More generally, we consider a technology ``implementable'' or ``profitable'' if there is some equilibrium (also in randomized sense -- Section~\ref{sec:airdop-logit})  which yields a positive profit to the designer. 


\subsection{Special Cases}
We shall sometime consider the following relevant restrictions on the technology functions, the possible contribution levels, the possible costs, and combinations thereof. 

\paragraph{Anonymous Technologies.}
    It is natural to assume that the value of the system simply depends on the overall level of contribution $\ell = \sum_i a_i$, that is, $V(a)=V(a')$ whenever $\sum_i a_i = \sum_i a_i'$. We refer to this case as the \emph{anonymous} technology function. With slight abuse of notation, we write $V(\ell)$ instead of $V(a)$.

\paragraph{Binary Contributions.}
    In some settings where potential contributors have only two options (strategies), either to contribute ($a_i=1$) or to not contribute ($a_i=0$), we refer to this restriction as \emph{binary contributions}.

\paragraph{Uniform Costs.}
    It is natural to consider equal cost for all (e.g. if they are fully determined by the type of hardware/resources required). This means $c_i = \alpha$, with $\alpha > 0$.

\section{Airdrop games: Main characteristics}\label{sec:airdrop-general}
In this section, we consider games with utilities in \eqref{eq:utiliy}, airdrop allocations \eqref{eq:rewards} in full generality. We show that these games are always potential games (Theorem~\ref{th:potential-game}), implying that (i) best/better response dynamics always converge to pure Nash equilibria and (ii) logit dynamics equilibria can be characterized in terms of potential, implying that bad equilibria are not selected under vanishing noise (Section~\ref{sec:airdop-logit}).

\begin{theorem}\label{th:potential-game}
	For airdrop allocation \eqref{eq:rewards}, the game in \eqref{eq:utiliy} with arbitrary efforts and any technology function  is an exact potential game with potential function 
 \begin{align}\label{eq:unifrom-rew-potential}
		\phi(a) := \gamma \cdot t(a) - SC(a) = \frac{\rho}{n} \cdot V(a) - SC(a)\ .
	\end{align}
\end{theorem}

\subsection{Characterization of pure Nash equilibria}\label{sec:aridrop-pne}

The next theorem characterizes the set of pure Nash equilibria. Intuitively, airdrop allocations should be neither too high (otherwise players can benefit by increasing their contribution) nor too low (otherwise they can benefit by reducing their contribution). Note an important \emph{asymmetry} between the two cases: increasing the contribution is never beneficial when the system's value does not change, while decreasing the contribution is always advantageous under these conditions.

\begin{theorem}\label{th:Nash-equilibria}
    For any technology function and  arbitrary efforts, and for airdrop allocations \eqref{eq:rewards}, the set of pure Nash equilibria is given by the strategy profiles $a$ such that the following two conditions hold for all $i$:
    \begin{enumerate}
        \item For all  $a^+ = (a_i^+,a_{-i})$ with  $a_i^{+} > a_i$: \newline
        $
            \frac{\rho}{n} \leq  c_i \cdot \frac{a_i^+ -a_i}{V(a^+) - V(a)}  $ whenever  $V(a) < V(a^+)$ .
        \item For all  $a^- = (a_i^-,a_{-i})$ with  $a_i^{-} < a_i$: \newline
        $
        \frac{\rho}{n}  \geq  c_i \cdot \frac{a_i -a_i^-}{V(a) - V(a^-)} $ and $V(a) > V(a^-) $ .
    \end{enumerate}
\end{theorem}


The next corollary provides a more convenient characterization for certain restrictions of interest.

\begin{corollary}
     For any technology function and  binary efforts, and for airdrop allocations \eqref{eq:rewards}, the set of pure Nash equilibria is given by the strategy profiles $a$ such that the following two conditions hold:
      \begin{align}
     \text{For all $i$ such that } a_i=0: && 
          \frac{\rho}{n} & \leq    \frac{c_i}{V(1,a_{-i}) - V(a)} \ \ \text{ or } \ \ V(1,a_{-i}) = V(a)   \ ;
          \\
         \text{For all $j$ such that }  a_j = 1: && 
          \frac{\rho}{n}  & \geq   \frac{c_j}{V(a) - V(0,a_{-j})}  \ \  \text{ and } \ \ V(0,a_{-j}) < V(a) \ . 
     \end{align}
    Moreover, for any anonymous technology function and binary efforts,  the set of strategy profiles $a$ which are an  equilibrium  corresponds to those satisfying these two conditions:
    \begin{enumerate}
    \item For $\ell < n$:  
         \begin{align}
         	\frac{\rho}{n}  \leq    \frac{c_{\min}^{(0)}}{V(\ell+1) - V(\ell)} \ \ \text{ or } \ \ V(\ell+1)= V(\ell) ,
         \end{align}
     where $c_{\min}^{(0)} = \min \{c_i : \ a_i = 0\}$ is the smallest cost among players not contributing.
     \item For  $\ell >0$:  
    \begin{align}
    	\frac{\rho}{n}   \geq   \frac{c_{\max}^{(1)}}{V(\ell) - V(\ell-1)} \ \ \text{ and } \ \ V(\ell-1) < V(\ell)  ; 
    \end{align}
    where $c_{\max}^{(1)} = \max \{c_i : \ a_i = 0\}$ is the largest cost among players contributing.
	\end{enumerate}
\end{corollary}

\begin{example}[linear technology with heterogeneous costs]\label{ex:linear-heterogeneous}
    For  linear technologies, $V(\ell) = \lambda_V \cdot \ell$ with $\lambda_V>0$,  the set of pure Nash equilibria is characterized by the $\ell^*$ such that (w.l.o.g. assume $c_1\leq c_2\leq \cdots \leq c_n$) these inequalities hold:
    $$
        \frac{c_{\ell^*}}{\lambda_V} \leq \frac{\rho}{n} \leq \frac{c_{\ell^*+1}}{\lambda_V} \ . 
    $$
    The optimal strategy (profit maximizing) for the designer is to choose the minimum $\rho$ satisfying the above condition, thus making the first inequality tight:
        $$\rho^* = \argmax_{\ell \in [n]} (1-\rho^*_\ell) \cdot \lambda_V \cdot \ell \ ,   \rho^*_\ell = n \cdot \frac{c_{\ell^*}}{\lambda_V} ,$$ which is equivalent to set $ 
        \rho^* = \argmax_{\ell \in [n]} (\lambda_V-n\cdot c_{\ell}) \cdot  \ell $ .
\end{example}

\subsection{Logit response equilibria}\label{sec:airdop-logit}
We next consider logit dynamics and the corresponding equilibria. Specifically, for the vanishing noise regime ($\beta\rightarrow \infty$), the dynamics selects a subset of  \emph{stochastically stable} pure Nash equilibria. 

\begin{theorem}\label{th:stable-equilibria}
For any technology function and  arbitrary efforts, for airdrop allocations \eqref{eq:rewards}, and  for vanishing noise ($\beta \rightarrow \infty$), the dynamics converges with probability which tends to $1$ to a state of maximal potential.~\footnote{Note that in this work we do not change sign in the definition of exact potential game and related dynamics. }   In particular, the corresponding stationary distribution $\pi_\rho$, depending on the airdrop allocation $\rho$, satisfies
\begin{align}
\lim_{\beta \rightarrow \infty}
    \left(\pi_\rho(a) \right)= \begin{cases}
        \frac{1}{|\potmax_\rho|} & \text{ for } a \in \potmax_\rho \\
        0 & \text{ for } a \not \in \potmax_\rho
    \end{cases} 
\end{align}
where $\potmax_\rho:= \argmax_{a}  \{\phi(a)\} $ is a subset of equilibria which depends on the airdrop allocation $\rho$ as follows: 
\begin{align}
    \potmax_\rho = \argmax_{a} \left\{\frac{\rho}{n}\cdot V(a) - SC(a)\right\} \ . 
\end{align}
\end{theorem}


The next example shows that, the bad equilibria in which no player contributes are selected with \emph{vanishing probability}, provided the airdrop allocation $\rho$ is sufficiently high.

\begin{example}[rule out bad equilibria]\label{ex:no-bad-eq}
	We consider a two-player game with a simple anonymous technology function,  binary contributions,  and  uniform costs ($c_i = \alpha$). The technology function is an \emph{AND} technology where a high (non-zero) value is achieved only if both players contribute, 
	$
		 V(0) = V(1) = 0 $ and $ V(2) > 0, 
	$
which implies that $(0,0)$ and $(1,1)$ are the only two Nash equilibria. 
Theorem~\ref{th:stable-equilibria} says that  logit-response dynamics with vanishing noise ($\beta\rightarrow \infty$)  the good equilibrium $(1,1)$ is reached with  probability tending to $1$ if and only if the airdrop allocation $\rho$ satisfies
 $\frac{\rho}{2} \cdot (V(2) - V(0)) > 2\alpha$.  
Also note that, since $\rho \leq 1$,  this holds only for $V(2) - V(0) > 4\alpha$. 
\end{example}

To repeat the intuition here, when contributors are prone to some experimentation, instead of just picking best responses, the system is likely to end up in the high value equilibrium instead of the low one.

\subsubsection{Profit and optimal airdrops revisited (logit dynamics)} We observe the same tradeoff regarding the optimal choice of airdrop  $\rho$. A too small $\rho$ may result in a bad equilibrium in which none cooperates, thus a system with small value $t(0)$. A very high $\rho$, on the other hand, will leave the system designer with only a tiny fraction of the system ownership (value). 
In order to deal with the logit (randomized) equilibrium $\pi_\rho$ we simply consider the expected value of the system, and extend the definition of profit \eqref{eq:profit} in the natural way:
\begin{align}
    \profit(\rho) = (1-\rho)V(\pi_\rho) - d_V ,  
    &&  \text{ where } V(\pi_\rho) := \EE{a\sim \pi_\rho}[V(a)]  \ . 
    \label{eq:profit-expected}
\end{align}

\section{Binary efforts and uniform costs}\label{sec:binary}
We consider a \emph{binary effort} scenario where players contribute or not ($a_i\in \{0,1\}$) and costs are \emph{uniform} ($c_i=\alpha$ for all $i$). 
Under these restrictions, our model is characterized by three parameters:
(i) $\alpha$ is the cost per player when contributing;
  (ii) $\beta$ is the rationality level of the players;
    (iii) $\rho$ is the airdrop allocation -- the corresponding number of tokens $\gamma$  is given by \eqref{eq:rewards}.
While the rationality level $\beta$ is exogenous to the system, the designer can change $\rho$, and costs  $\alpha$ are part of the airdrop design (see \appsec~\ref{sec:partnerchain} for practical examples where the designer can reduce $\alpha$).  

\subsection{Logit equilibria and convergence time}\label{sec:logit-time}
In this section, we make use of \emph{birth and death processes} to analyze the case of binary effort and airdrop rewards, when also the technology function is \emph{symmetric}. In this case, the dynamics boils down to a birth and death process,  where we keep track of the number $\ell$ of actively participating players in a given profile $a$, i.e., $\ell = \sum_i a_i$ and $a_i\in\{0,1\}$.  Hence, the birth and death process has $n+1$ states, $\ell \in [0,n]$. The stationary distribution is thus
\begin{align}\label{eq:birth-and-death-stationary}
    \hat{\pi}(\ell) 
     := \binom{n}{\ell}\cdot \pi(\ell) 
\end{align}
where $\pi(\ell)$ is the stationary distribution of a generic  state $a$ with $\sum_i a_i = \ell$, and the binomial coefficient counts the number of such states in the original Markov chain which are ``grouped together'' in the birth and death process. 
The birth and death process has transition probabilities $p(\ell)$ and $q(\ell)$ of moving ``up by one''  or ``down by one'', respectively, given by
\begin{align}\label{eq:birth-death-rate-general}
    p(\ell) = \frac{n-\ell}{n} \cdot p_i^{\beta}(1|a_{-i})\ ,  && q(\ell) = \frac{\ell}{n} \cdot p_i^{\beta}(0|a_{-i}) 
\end{align}
where $p_i^{\beta}(\cdot)$ is the logit response  \eqref{eq:logit-response}.

Since our process in  \eqref{eq:birth-death-rate-general} is an irreducible birth and death chain, the following ``sharp'' bound on the  \emph{mixing time} holds. 

\newcommand{\Time}{\mathcal{T}}
\newcommand{\tmix}{\Time_{\mathsf{mix}}}
\newcommand{\tcut}{\Time_{\mathsf{cutoff}}}
\newcommand{\thit}{\Time_{\mathsf{hitting}}}

\begin{theorem}
    [Theorem~1.1 in \cite{chen2013mixing}]\label{th:mixing-birth-and-death}
Let $\ell_0$ be a state satisfying $\hat{\pi}([0, \ell_0]) \geq  1/2$
and $\hat{\pi}([\ell_0, n]) \geq 1/2$, where $\hat{\pi}(I) := \sum_{\ell \in I} \hat{\pi}(\ell)$, and set
    \begin{align}
        \tcut = \max\left\{\sum_{\ell = 0}^{\ell_0-1} \frac{\hat{\pi}([0,\ell])}{\hat{\pi}(\ell)p(\ell)}, \sum_{\ell = \ell_0 + 1}^{n} \frac{\hat{\pi}([\ell,n])}{\hat{\pi}(\ell)q(\ell)}\right\} \ . 
    \end{align} 
    Then the mixing time of the logit dynamics satisfies
    $\tmix = \Theta(\tcut)$  and, in particular, the following bounds hold:\begin{align}\label{eq:tmix-general}
    (1/24) \cdot \tcut \leq \tmix \leq 288 \cdot \tcut \ . 
    \end{align}
\end{theorem}

Note that such an $\ell_0$ always exists (sum up all $\hat{\pi}(\ell)$ from $0$ until the smallest $\ell_0$ where the sum of these probabilities is at least $1/2$).  We next  introduce useful definitions to analyze the \emph{hitting time} of a target $\ell$, based on the technology's ``local steepness''. 

\begin{definition}\label{def:drift}
    We define the drift at location $\ell$ as the ratio $d(\ell):= \hat{\pi}(\ell+1)/\hat{\pi}(\ell)$. We also say that technology function $V$ is $s$-steep at some interval $I=[\ell_1,\ell_2]$ if $V(\ell+1) - V(\ell) \leq   s$ for all $\ell \in I\setminus \{\ell_2\}$.  
\end{definition}

Intuitively, the drift describes the tendency of the process to ``move down'' (drift $< 1$) or ``move up'' (drift $> 1$). The following theorem states that the hitting time for a target value $\ell$ grows exponentially with the length of any interval, preceding the target value, where the tendency to move down persists. The theorem further connects the drift to the ``flatness'' of the technology function (see the threshold function in Section~\ref{sec:threshold}).
Its proof is based on bounds in  \cite{PALACIOS1996119}.

\begin{theorem}\label{thm:hitting-time-general}The hitting time $\thit(\ell)$ of the logit dynamics to reach a state with contribution level $\ell$ starting tom the state with contribution level $\ell=0$ can be bounded as follows: 
\begin{enumerate}
    \item If the drift in some interval $I=[\ell_1,\ell_2]$ is at most $d_I$, then  $\thit(\ell) \geq (1/d_I)^{|I|}$ for all $\ell > \ell_2$. 
    \item If $V$ is $s$-steep in some interval $I=[\ell_1,\ell_2]$, then for all $\ell > \ell_2$ it holds that 
    $$
        \thit(\ell) \geq \left(\exp\left(-\beta \left(\frac{\rho}{n} \cdot s - \alpha\right)\right)\cdot \frac{\ell_1+1}{n-\ell_1} \right)^{\ell_2 - \ell_1} .$$ 
\end{enumerate}
\end{theorem}

\section{Application to threshold technologies}\label{sec:threshold}
We analyze  a \emph{threshold technology}, modelling scenarios in which the system is either highly valuable if the overall contribution of the players reaches a certain threshold $\tau$, and less valuable otherwise:
\begin{align}\label{eq:threshold-t}
    V(\ell) = \begin{cases}
        \low{V} & \ell < \tau \\
        \high{V} & \ell \geq  \tau
    \end{cases} && \low{V} < \high{V} \ ,
\end{align}
where $\ell$ is the number of actively participating players, i.e., $\ell = \sum_i a_i$ and $a_i\in\{0,1\}$. The corresponding token values according to \eqref{eq:token-val} are thus $\low{t} =\low{V}/T_{tot}$ and 
$\high{t} =\high{V}/T_{tot}$. We are interested in the probability that the underline dynamics selects the high value (optimal) outcome, 
\begin{align}\label{eq:threshold-prob-opt}
        \high{p}(\rho):= \Pr_{a\sim \pi_\rho}[V(a) = \high{V}] 
        \ . 
    \end{align}

\subsection{Stochastic stability ($\beta \rightarrow \infty$ regime)}\label{sec:threshold-vanishing}

\begin{theorem}\label{th:pot-max-threshold}
    For any threshold technology \eqref{eq:threshold-t}  with airdrop rewards \eqref{eq:rewards} and  vanishing  noise ($\beta\rightarrow \infty$), the probability of selecting the high value outcome \eqref{eq:threshold-prob-opt} undergoes  a sharp transition given by the rewards $\rho$:
    \begin{align}\label{eq:rho-critical-threshold}
        \lim_{\beta \rightarrow \infty} \high{p}(\rho) = 
        \begin{cases}
            1 & \rho > \rho_c\\ 
            0 & \rho < \rho_c
        \end{cases},  && 
        \rho_c :=  \frac{\alpha \cdot n \cdot \tau}{\high{V} - \low{V}} \ . 
    \end{align}
    For the edge case where $\rho=\rho_c$, the probability satisfies $\lim_{\beta \rightarrow \infty} \high{p}(\rho) = 1/\left(1+\binom{n}{\tau}\right).$ 
\end{theorem}


An immediate corollary of the previous result follows. Intuitively, the corollary states that there exists three regions: (i) for very high cost, the probability of selecting the good outcome vanishes no matter how we set the rewards, and thus the optimal strategy of the designer is to set no airdrop, (ii) for intermediate costs, though it is possible to set $\rho>0$ such that the probability of selecting the good outcome tends to one, the designer still prefer to set $\rho=0$, and (iii) for low costs there is $\rho>0$ maximizing the designer's profit and making the probability of selecting the good outcome going to one. 

\begin{corollary}\label{cor:stable-theshold}
    For any threshold technology \eqref{eq:threshold-t}  with airdrop rewards \eqref{eq:rewards} and  vanishing  noise ($\beta\rightarrow \infty$), the probability of selecting the high value outcome \eqref{eq:threshold-prob-opt} is as follows:
    \begin{enumerate}
        \item For $\alpha\cdot n \cdot \tau> \high{V}-\low{V}$ the probability of selecting the high value outcome vanishes for any airdrop reward $\rho$. Hence, and the best strategy for the designer is to give no airdrop rewards, which guarantees a profit of $\low{V}-d_V$.
        \item  For $\alpha\cdot n \cdot \tau < \high{V}-\low{V}$ the probability of selecting the high value outcome tends to $1$ for any airdrop reward $\rho>\rho_c$. The optimal strategy (profit maximizing) for the designer is as follows: 
        \begin{enumerate}
            \item \label{cor:stable-threshold:intermediate}For $\alpha\cdot n \cdot \tau \geq  (\high{V}-\low{V}) \cdot (1 - \low{V}/\high{V})$ it is still optimal for the designer to give no rewards (causing the probability of selecting the good outcome to vanish).
            \item For $\alpha\cdot n \cdot \tau <  (\high{V}-\low{V}) \cdot (1 - \low{V}/\high{V})$ the best strategy for the designer is to set airdrop rewards slightly above $\rho_c<1$, which guarantees a profit of $(1-\rho_c - \epsilon)\high{V}-d_V$ for any small $\epsilon>0$.
        \end{enumerate}
    \end{enumerate}
\end{corollary}
Note that the ``intermediate'' regime in part~\ref{cor:stable-threshold:intermediate} of the corollary above occurs only for $\low{V}>0$. Here high rewards \emph{could} make the system succeed, but they are not optimal for the designer.    
For $\low{V}=0$ we have a single transition (either provide no airdrop or set the airdrop slightly above $\rho_c$).

\subsection{Non-vanishing noise ($\beta$ finite regime)}\label{sec:threshold-nonvanishing}
In this section, we analyze logit dynamics for threshold technologies in the case of non-vanishing inverse noise $\beta>0$. Research suggest that in practice people respond according to some specific value of $\beta$ which is approximately the same across different games and situations they face (see \appsec~\ref{sec:logit-values}).  The next result provides useful bounds on the probability that the high value outcome is selected at equilibrium by the dynamics.  

\begin{theorem}\label{th:threshold-finite-beta:prob}
    For any threshold technology \eqref{eq:threshold-t}  with airdrop rewards \eqref{eq:rewards} and  any inverse noise parameter $\beta>0$, the probability of selecting the high value outcome \eqref{eq:threshold-prob-opt} is monotone increasing in the rewards $\rho$ and, in particular, it has the following form: 
    \begin{align*}
        \high{p} (\rho)
        = \frac{1}{1+C \cdot \exp(-\rho B) } ,   && B = \frac{\beta}{n} \cdot (\high{V}- \low{V}) \ , 
    \end{align*}
    where $C=C(\alpha\beta,n,\tau)=\frac{1 - \high{p}(0)}{\high{p}(0)}$ does not depend on rewards $\rho$ nor on the values $\low{V}$ and $\high{V}$.
\end{theorem}

Based on the result above, we are able to characterize the optimal airdrop rewards for the designer. 

\begin{theorem}\label{th:threshold:profit-logit}
    For any threshold technology \eqref{eq:threshold-t} with $\low{V}=0$, and with airdrop rewards \eqref{eq:rewards} and  any inverse noise parameter $\beta>0$, the designer's profit \eqref{eq:profit-expected} is  
    \begin{align}
    \label{eq:expected-profit-threshold}
        \profit(\rho)  
        & = \high{V} \cdot \frac{1-\rho}{1 + C \cdot \exp(-\rho \cdot B)} - d_V\ ,  
    \end{align}
        where quantities $C$ and $B=\frac{\beta}{n}\cdot \high{V}$ are defined as in Theorem~\ref{th:threshold-finite-beta:prob}. Moreover the following holds:
    \begin{enumerate}
        \item For $n \geq \beta \cdot \high{V}$ the optimal strategy (profit maximizing) for the designer is to give no airdrop rewards, which guarantees a profit of $\high{V}\cdot \high{p}(0)-d_V$. 
        \item For $n < \beta \cdot \high{V}$ the optimal strategy (profit maximizing) for the designer is to set an airdrop reward $\rho \leq \bar{\rho} := 1 - 1/B = 1 - \frac{n}{\beta \high{V}}$. The probability of selecting the high value outcome for the designer's optimal rewards $\rho^*$ is  bounded as follows:
        $ \high{p}(\rho^*)  \leq \high{p}(\bar{\rho})  =    \frac{1}{1+C\cdot \exp(1-B)} = \frac{1}{1+C \cdot \exp(1 - \beta \cdot \high{V}/n) }$. 
        
        \item For $n < \beta \cdot \high{V}\cdot (1 - \high{p}(0))$ the optimal strategy (profit maximizing) for the designer is to give  strictly positive airdrop rewards.
    \end{enumerate}
\end{theorem}

\begin{figure}
    \centering
   \includegraphics[scale = .6]{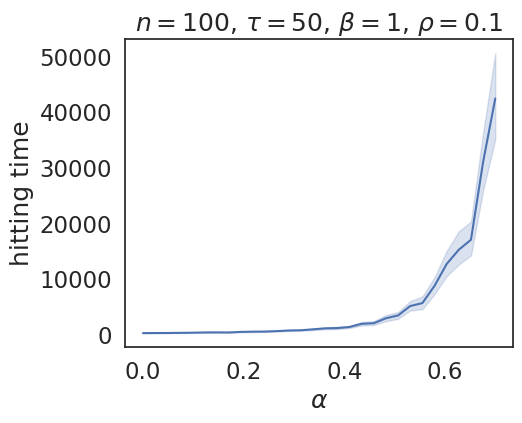}
   \includegraphics[scale = .6]{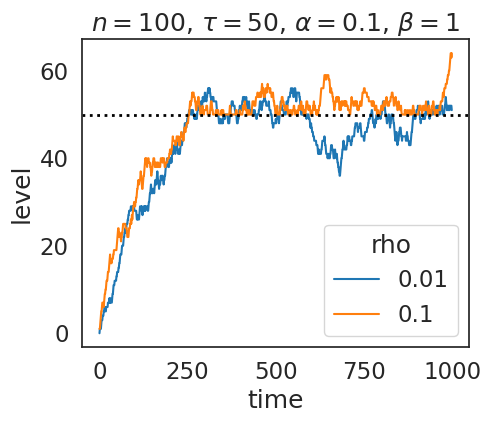}

     \caption{On the left, larger costs $\alpha$ increase the hitting time  (100 repetitions 95\% confidence). On the right, larger  rewards values $\rho$ help to maintain the dynamics above the threshold once it is reached.}
     \label{fig:dynamics-time}
\end{figure}

\subsection{Convergence time}\label{sec:threshold-time} In this section, we study the \emph{time} for logit dynamics to converge  to its equilibrium (stationary distribution) and to the good outcome of the threshold function ($\ell \geq \tau$). Specifically, we provide tight bounds on the \emph{mixing time} (Theorem~\ref{th:mixing-LB-threshold})  and  on the \emph{hitting time} of a target value (Theorem~\ref{th:hitting-threshold-UB-LB}). 

\paragraph{Some intuition first.}
We observe experimentally (Figure~\ref{fig:dynamics-time}) that lower costs $\alpha$ \emph{accelerates}  convergence to the desired ``high value'' region, while increasing  rewards $\rho$  helps to \emph{maintain} the desired equilibrium (but it does \emph{not} accelerate  convergence). Intuitively,  the dynamics converge quickly to an ``average''  contribution level $\ell^*$ which depends \emph{only} on $\alpha\beta$:\begin{align}\label{eq:equilibrium-alphabeta}
        \ell^*  = n \cdot p_{\alpha\beta}\ ,  && p_{\alpha\beta} := \frac{1}{1 + \exp(\alpha\beta)} \ . 
    \end{align} Then, convergence to the desired ``high value'' region $\ell \geq \tau$ is fast for $\tau \leq \ell^*$ but becomes \emph{slow} for $\tau > \ell^*$. This suggests that the convergence time grows with the \emph{gap} $\tau - \ell^*$ and the hard case is when $\ell^* \ll \tau$.

\paragraph{Formal analysis}
As for the mixing time, we leverage on the bounds for birth and death chains (Theorem~\ref{th:mixing-birth-and-death}). To this end, we note that in this particular case of threshold functions, the birth and death probabilities \eqref{eq:birth-death-rate-general} assume a special  form. This leads to the next theorem, which  provides a lower bound on the mixing time for the ``useful'' scenario, that is, when the success probability is larger than the failure probability. 

\begin{theorem}[mixing time]\label{th:mixing-LB-threshold}
 For any threshold technology \eqref{eq:threshold-t}  and  airdrop rewards \eqref{eq:rewards},  if $\high{p}(\rho)>1/2$, then the mixing time can be bounded as follows:  $\tmix = \Theta(\tcut)$ and 
    \begin{align}
    \tcut \geq \sum_{\ell = 0}^{\tau} \frac{\hat{\pi}([0,\ell])}{\hat{\pi}(\ell)p(\ell)} \geq \exp(\alpha\beta) \cdot  \frac{\exp(\tau-1)}{\binom{n}{\tau-1}} \ . 
    \end{align} 
\end{theorem}

We next provide bounds on the \emph{hitting time} of the good value, that is, to reach a contribution level $\ell = \tau$. The first part of the next theorem says that the dynamics converge quickly to a contribution level $\ell^*$ given by \eqref{eq:equilibrium-alphabeta}. 

\begin{theorem}[hitting time]\label{th:hitting-threshold-UB-LB}
    For any threshold technology \eqref{eq:threshold-t}  and  airdrop rewards \eqref{eq:rewards}, let $\thit(\ell)$ be the expected time  for the logit dynamics to reach state $\ell$ starting from state $0$. Then, for $\ell^*$ defined as in \eqref{eq:equilibrium-alphabeta},  the following holds:
    \begin{enumerate}
        \item (Upper Bound). \label{th:hitting-threshold-UB-itm} 
     $   \thit(\ell^*) \leq O\left(n^2 \cdot  \frac{\ell^*}{n - \ell^*}\right)$. 
    \item (Lower Bound). \label{th:hitting-threshold-LB-itm}
     $   \thit(\tau) \geq \left(\exp(\alpha\beta)\cdot \frac{\ell+1}{n-\ell}\right)^{\tau-\ell }$, 
    for all  $0 \leq \ell \leq \tau$.  This  implies, 
    $
        \thit(\tau) \geq (1+1/\ell^*)^{\tau-\ell^* -1 }  
    $.
    \end{enumerate}
\end{theorem}

The second part of the theorem states that the time to reach the ``high value'' region increases with larger $\alpha\beta$, growing exponentially with the gap between $\ell^*$ and $\tau > \ell^*$.

\section{Conclusions and future work}

This paper presents a game-theoretic framework to address the dynamics involved in launching a new blockchain, specifically focusing on how contributions can be incentivized through token rewards and the possibility of support from an established mainchain. The analysis provided, offers both theoretical and practical insights into the design of blockchain launches, airdrop mechanisms, and the use of mainchain resources to achieve successful outcomes.

\bibliography{tokenomics,random_process_tools,crowdfunding,airdrops}

\begin{thebibliography}{62}
\providecommand{\natexlab}[1]{#1}
\providecommand{\url}[1]{\texttt{#1}}
\expandafter\ifx\csname urlstyle\endcsname\relax
  \providecommand{\doi}[1]{doi: #1}\else
  \providecommand{\doi}{doi: \begingroup \urlstyle{rm}\Url}\fi

\bibitem[Alabi(2017)]{ALABI201723}
Ken Alabi.
\newblock Digital blockchain networks appear to be following {Metcalfe’s
  Law}.
\newblock \emph{Electronic Commerce Research and Applications}, 24:\penalty0
  23--29, 2017.
\newblock ISSN 1567-4223.
\newblock \doi{https://doi.org/10.1016/j.elerap.2017.06.003}.
\newblock URL
  \url{https://www.sciencedirect.com/science/article/pii/S1567422317300480}.

\bibitem[Allen(2024)]{evolutionary}
Darcy~WE Allen.
\newblock Crypto airdrops: An evolutionary approach.
\newblock \emph{Journal of Evolutionary Economics}, pages 1--24, 2024.

\bibitem[Allen et~al.(2023)Allen, Berg, and Lane]{why}
Darcy~WE Allen, Chris Berg, and Aaron~M Lane.
\newblock Why airdrop cryptocurrency tokens?
\newblock \emph{Journal of Business Research}, 163:\penalty0 113945, 2023.

\bibitem[Al{\'o}s-Ferrer and Netzer(2010)]{alos2010logit}
Carlos Al{\'o}s-Ferrer and Nick Netzer.
\newblock The logit-response dynamics.
\newblock \emph{Games and Economic Behavior}, 68\penalty0 (2):\penalty0
  413--427, 2010.

\bibitem[Al{\'o}s-Ferrer and Netzer(2015)]{alos2015robust}
Carlos Al{\'o}s-Ferrer and Nick Netzer.
\newblock Robust stochastic stability.
\newblock \emph{Economic Theory}, 58:\penalty0 31--57, 2015.

\bibitem[Al{\'o}s-Ferrer and Netzer(2017)]{alos2017convergence}
Carlos Al{\'o}s-Ferrer and Nick Netzer.
\newblock On the convergence of logit-response to (strict) nash equilibria.
\newblock \emph{Economic Theory Bulletin}, 5:\penalty0 1--8, 2017.

\bibitem[Arieli et~al.(2018)Arieli, Koren, and Smorodinsky]{crwodfundinggame}
Itai Arieli, Moran Koren, and Rann Smorodinsky.
\newblock The one-shot crowdfunding game.
\newblock In \emph{Proc. of the ACM Conference on Economics and Computation
  (EC)}, page 213–214, 2018.
\newblock ISBN 9781450358293.
\newblock \doi{10.1145/3219166.3219215}.
\newblock URL \url{https://doi.org/10.1145/3219166.3219215}.

\bibitem[Asadpour and Saberi(2009)]{asadpour2009inefficiency}
Arash Asadpour and Amin Saberi.
\newblock On the inefficiency ratio of stable equilibria in congestion games.
\newblock In \emph{Proc. of the 5th International Workshop on Internet and
  Network Economics (WINE)}, pages 545--552. Springer, 2009.

\bibitem[Auletta et~al.(2011)Auletta, Ferraioli, Pasquale, Penna, and
  Persiano]{auletta2011convergence}
Vincenzo Auletta, Diodato Ferraioli, Francesco Pasquale, Paolo Penna, and
  Giuseppe Persiano.
\newblock Convergence to equilibrium of logit dynamics for strategic games.
\newblock In \emph{Proc. of the 23rd annual ACM Symposium on Parallelism in
  Algorithms and Architectures (SPAA)}, pages 197--206, 2011.

\bibitem[Auletta et~al.(2012)Auletta, Ferraioli, Pasquale, and
  Persiano]{auletta2012metastability}
Vincenzo Auletta, Diodato Ferraioli, Francesco Pasquale, and Giuseppe Persiano.
\newblock Metastability of logit dynamics for coordination games.
\newblock In \emph{Proc. of the 23rd annual ACM-SIAM Symposium on Discrete
  Algorithms (SODA)}, pages 1006--1024. SIAM, 2012.

\bibitem[Auletta et~al.(2013{\natexlab{a}})Auletta, Ferraioli, Pasquale, Penna,
  and Persiano]{auletta2013logit}
Vincenzo Auletta, Diodato Ferraioli, Francesco Pasquale, Paolo Penna, and
  Giuseppe Persiano.
\newblock Logit dynamics with concurrent updates for local interaction games.
\newblock In \emph{Proc. of the 21st Annual European Symposium on Algorithms
  (ESA)}, pages 73--84. Springer, 2013{\natexlab{a}}.

\bibitem[Auletta et~al.(2013{\natexlab{b}})Auletta, Ferraioli, Pasquale, and
  Persiano]{auletta2013mixing}
Vincenzo Auletta, Diodato Ferraioli, Francesco Pasquale, and Giuseppe Persiano.
\newblock Mixing time and stationary expected social welfare of logit dynamics.
\newblock \emph{Theory of Computing Systems}, 53:\penalty0 3--40,
  2013{\natexlab{b}}.

\bibitem[Babaioff et~al.(2012)Babaioff, Feldman, Nisan, and
  Winter]{BABAIOFF2012999}
Moshe Babaioff, Michal Feldman, Noam Nisan, and Eyal Winter.
\newblock Combinatorial agency.
\newblock \emph{Journal of Economic Theory}, 147\penalty0 (3):\penalty0
  999--1034, 2012.
\newblock ISSN 0022-0531.
\newblock \doi{https://doi.org/10.1016/j.jet.2012.01.010}.
\newblock URL
  \url{https://www.sciencedirect.com/science/article/pii/S0022053112000117}.

\bibitem[Bil{\`{o}} et~al.(2023)Bil{\`{o}}, Gourv{\`{e}}s, and
  Monnot]{DBLP:journals/tcs/BiloGM23}
Vittorio Bil{\`{o}}, Laurent Gourv{\`{e}}s, and J{\'{e}}r{\^{o}}me Monnot.
\newblock Project games.
\newblock \emph{Theor. Comput. Sci.}, 940\penalty0 (Part):\penalty0 97--111,
  2023.
\newblock \doi{10.1016/J.TCS.2022.10.043}.
\newblock URL \url{https://doi.org/10.1016/j.tcs.2022.10.043}.

\bibitem[Blume(1993)]{blume1993statistical}
Lawrence~E Blume.
\newblock The statistical mechanics of strategic interaction.
\newblock \emph{Games and economic behavior}, 5\penalty0 (3):\penalty0
  387--424, 1993.

\bibitem[Blume(2003)]{BLUME2003251}
Lawrence~E. Blume.
\newblock How noise matters.
\newblock \emph{Games and Economic Behavior}, 44\penalty0 (2):\penalty0
  251--271, 2003.
\newblock ISSN 0899-8256.
\newblock \doi{https://doi.org/10.1016/S0899-8256(02)00554-7}.
\newblock URL
  \url{https://www.sciencedirect.com/science/article/pii/S0899825602005547}.

\bibitem[Boncinelli and Pin(2012)]{boncinelli2012stochastic}
Leonardo Boncinelli and Paolo Pin.
\newblock Stochastic stability in best shot network games.
\newblock \emph{Games and Economic Behavior}, 75\penalty0 (2):\penalty0
  538--554, 2012.

\bibitem[Bramoull{\'e} and Kranton(2007)]{bramoulle2007public}
Yann Bramoull{\'e} and Rachel Kranton.
\newblock Public goods in networks.
\newblock \emph{Journal of Economic theory}, 135\penalty0 (1):\penalty0
  478--494, 2007.

\bibitem[Buragohain et~al.(2003)Buragohain, Agrawal, and
  Suri]{buragohain2003game}
Chiranjeeb Buragohain, Divyakant Agrawal, and Subhash Suri.
\newblock A game theoretic framework for incentives in {P2P} systems.
\newblock In \emph{Proc. of the 3rd International Conference on Peer-to-Peer
  Computing (P2P)}, pages 48--56. IEEE, 2003.

\bibitem[Chaidos et~al.(2023)Chaidos, Kiayias, and
  Markakis]{chaidos2023blockchain}
Pyrros Chaidos, Aggelos Kiayias, and Evangelos Markakis.
\newblock Blockchain participation games.
\newblock In \emph{Proc. of the 19th International Conference on Web and
  Internet Economics (WINE)}, pages 169--187. Springer, 2023.

\bibitem[Chang(2020)]{chang2020economics}
Jen-Wen Chang.
\newblock The economics of crowdfunding.
\newblock \emph{American Economic Journal: Microeconomics}, 12\penalty0
  (2):\penalty0 257--280, 2020.

\bibitem[Chen and Saloff-Coste(2013)]{chen2013mixing}
Guan-Yu Chen and Laurent Saloff-Coste.
\newblock On the mixing time and spectral gap for birth and death chains.
\newblock \emph{ALEA-Latin American Journal of Probability and Mathematical
  Statistics}, 10\penalty0 (1):\penalty0 293--321, 2013.

\bibitem[Costain and Nakov(2019)]{costain2019logit}
James Costain and Anton Nakov.
\newblock Logit price dynamics.
\newblock \emph{Journal of Money, Credit and Banking}, 51\penalty0
  (1):\penalty0 43--78, 2019.

\bibitem[Coucheney et~al.(2014)Coucheney, Durand, Gaujal, and
  Touati]{coucheney2014general}
Pierre Coucheney, St{\'e}phane Durand, Bruno Gaujal, and Corinne Touati.
\newblock General revision protocols in best response algorithms for potential
  games.
\newblock In \emph{7th Int. Conference on NETwork Games, COntrol and
  OPtimization (NetGCoop)}, pages 239--246. IEEE, 2014.

\bibitem[Crawford and Iriberri(2007)]{crawford2007level}
Vincent~P Crawford and Nagore Iriberri.
\newblock Level-k auctions: Can a nonequilibrium model of strategic thinking
  explain the winner's curse and overbidding in private-value auctions?
\newblock \emph{Econometrica}, 75\penalty0 (6):\penalty0 1721--1770, 2007.

\bibitem[Dall’Asta et~al.(2011)Dall’Asta, Pin, and
  Ramezanpour]{dall2011optimal}
Luca Dall’Asta, Paolo Pin, and Abolfazl Ramezanpour.
\newblock Optimal equilibria of the best shot game.
\newblock \emph{Journal of Public Economic Theory}, 13\penalty0 (6):\penalty0
  885--901, 2011.

\bibitem[Do~Dinh and Hollender(2024)]{do2024tight}
J{\'e}r{\'e}mi Do~Dinh and Alexandros Hollender.
\newblock Tight inapproximability of nash equilibria in public goods games.
\newblock \emph{Information Processing Letters}, 186:\penalty0 106486, 2024.

\bibitem[EigenLabs(2014)]{eigenlayer}
EigenLabs.
\newblock The universal intersubjective work token.
\newblock \emph{White paper}, 2014.
\newblock url: https://www.blog.eigenlayer.xyz/eigen/.

\bibitem[Fan et~al.(2023)Fan, Min, Wu, and Cai]{altruistic}
Sizheng Fan, Tian Min, Xiao Wu, and Wei Cai.
\newblock Altruistic and profit-oriented: Making sense of roles in web3
  community from airdrop perspective.
\newblock In \emph{Proc. of the Conference on Human Factors in Computing
  Systems (CHI)}. Association for Computing Machinery, 2023.
\newblock ISBN 9781450394215.
\newblock \doi{10.1145/3544548.3581173}.
\newblock URL \url{https://doi.org/10.1145/3544548.3581173}.

\bibitem[Ferraioli and Ventre(2017)]{ferraioli2017social}
Diodato Ferraioli and Carmine Ventre.
\newblock Social pressure in opinion games.
\newblock In \emph{Proc. of the 26th International Joint Conference on
  Artificial Intelligence (IJCAI)}, pages 3661--3667, 2017.

\bibitem[Ferraioli et~al.(2016)Ferraioli, Goldberg, and
  Ventre]{ferraioli2016decentralized}
Diodato Ferraioli, Paul~W Goldberg, and Carmine Ventre.
\newblock Decentralized dynamics for finite opinion games.
\newblock \emph{Theoretical Computer Science}, 648:\penalty0 96--115, 2016.

\bibitem[Fr{\"o}wis and B{\"o}hme(2019)]{frowis2019operational}
Michael Fr{\"o}wis and Rainer B{\"o}hme.
\newblock The operational cost of ethereum airdrops.
\newblock In \emph{Data Privacy Management, Cryptocurrencies and Blockchain
  Technology: ESORICS 2019 International Workshops, DPM 2019 and CBT 2019},
  pages 255--270. Springer, 2019.

\bibitem[Galeotti and Goyal(2010)]{10.1257/aer.100.4.1468}
Andrea Galeotti and Sanjeev Goyal.
\newblock The law of the few.
\newblock \emph{American Economic Review}, 100\penalty0 (4):\penalty0
  1468–92, September 2010.
\newblock \doi{10.1257/aer.100.4.1468}.
\newblock URL \url{https://www.aeaweb.org/articles?id=10.1257/aer.100.4.1468}.

\bibitem[Galeotti et~al.(2010)Galeotti, Goyal, Jackson, Vega-Redondo, and
  Yariv]{galeotti2010network}
Andrea Galeotti, Sanjeev Goyal, Matthew~O Jackson, Fernando Vega-Redondo, and
  Leeat Yariv.
\newblock Network games.
\newblock \emph{The review of economic studies}, 77\penalty0 (1):\penalty0
  218--244, 2010.

\bibitem[Georganas(2011)]{georganas2011english}
Sotiris Georganas.
\newblock English auctions with resale: An experimental study.
\newblock \emph{Games and Economic Behavior}, 73\penalty0 (1):\penalty0
  147--166, 2011.

\bibitem[Gilboa and Nisan(2022)]{gilboa2022complexity}
Matan Gilboa and Noam Nisan.
\newblock Complexity of public goods games on graphs.
\newblock In \emph{Proc. of the 15th International Symposium on Algorithmic
  Game Theory (SAGT)}, pages 151--168. Springer, 2022.

\bibitem[Goeree and Holt(2001)]{goeree2001ten}
Jacob~K Goeree and Charles~A Holt.
\newblock Ten little treasures of game theory and ten intuitive contradictions.
\newblock \emph{American Economic Review}, 91\penalty0 (5):\penalty0
  1402--1422, 2001.

\bibitem[Husain(2024)]{IOHK-PC}
Omer Husain.
\newblock Announcing the alpha v1 release of the partner chains toolkit.
\newblock \emph{IOHK blog post}, August 2024.
\newblock URL \url{https://iohk.io/en/blog/posts/2024/08/page-1/}.

\bibitem[Jiménez-Jiménez et~al.(2021)Jiménez-Jiménez, Alba-Fernández, and
  Martínez-Gómez]{math9212757}
Francisca Jiménez-Jiménez, Maria~Virtudes Alba-Fernández, and Cristina
  Martínez-Gómez.
\newblock Attracting the right crowd under asymmetric information: A game
  theory application to rewards-based crowdfunding.
\newblock \emph{Mathematics}, 9\penalty0 (21), 2021.
\newblock ISSN 2227-7390.
\newblock \doi{10.3390/math9212757}.
\newblock URL \url{https://www.mdpi.com/2227-7390/9/21/2757}.

\bibitem[Kleer(2021)]{kleer2021sampling}
Pieter Kleer.
\newblock Sampling from the gibbs distribution in congestion games.
\newblock In \emph{Proc. of the 22nd ACM Conference on Economics and
  Computation (EC)}, pages 679--680, 2021.

\bibitem[Klimm and Stahlberg(2023)]{10.1145/3580507.3597780}
Max Klimm and Maximilian~J. Stahlberg.
\newblock Complexity of equilibria in binary public goods games on undirected
  graphs.
\newblock In \emph{Proc. of the 24th ACM Conference on Economics and
  Computation (EC)}, page 938–955, New York, NY, USA, 2023. Association for
  Computing Machinery.
\newblock ISBN 9798400701047.
\newblock \doi{10.1145/3580507.3597780}.
\newblock URL \url{https://doi.org/10.1145/3580507.3597780}.

\bibitem[Levin et~al.(2006)Levin, Peres, and Wilmer]{LevinPeresWilmer2006}
David~A. Levin, Yuval Peres, and Elizabeth~L. Wilmer.
\newblock \emph{{Markov chains and mixing times}}.
\newblock American Mathematical Society, 2006.

\bibitem[Lommers et~al.(2023)Lommers, Makridis, and Verboven]{designing}
Kristof Lommers, Christos Makridis, and Lieven Verboven.
\newblock Designing airdrops.
\newblock \emph{Available at SSRN 4427295}, 2023.

\bibitem[Makridis et~al.(2023)Makridis, Fröwis, Sridhar, and Böhme]{rise}
Christos~A. Makridis, Michael Fröwis, Kiran Sridhar, and Rainer Böhme.
\newblock {The rise of decentralized cryptocurrency exchanges: Evaluating the
  role of airdrops and governance tokens}.
\newblock \emph{Journal of Corporate Finance}, 79\penalty0 (C), 2023.
\newblock \doi{10.1016/j.jcorpfin.2023.1}.
\newblock URL
  \url{https://ideas.repec.org/a/eee/corfin/v79y2023ics092911992300007x.html}.

\bibitem[Mamageishvili and Penna(2016)]{MamageishviliP16}
Akaki Mamageishvili and Paolo Penna.
\newblock Tighter bounds on the inefficiency ratio of stable equilibria in load
  balancing games.
\newblock \emph{Oper. Res. Lett.}, 44\penalty0 (5):\penalty0 645--648, 2016.
\newblock \doi{10.1016/J.ORL.2016.07.014}.
\newblock URL \url{https://doi.org/10.1016/j.orl.2016.07.014}.

\bibitem[Messias et~al.(2023)Messias, Yaish, and Livshits]{harder}
Johnnatan Messias, Aviv Yaish, and Benjamin Livshits.
\newblock Airdrops: Giving money away is harder than it seems.
\newblock \emph{arXiv preprint arXiv:2312.02752}, 2023.

\bibitem[Montanari and Saberi(2009)]{10.1109/FOCS.2009.64}
Andrea Montanari and Amin Saberi.
\newblock Convergence to equilibrium in local interaction games.
\newblock In \emph{Proc. of the 50th Annual IEEE Symposium on Foundations of
  Computer Science (FOCS)}, page 303–312. IEEE Computer Society, 2009.
\newblock ISBN 9780769538501.
\newblock \doi{10.1109/FOCS.2009.64}.
\newblock URL \url{https://doi.org/10.1109/FOCS.2009.64}.

\bibitem[M{\"u}ller et~al.(2021)M{\"u}ller, Nesterov, and
  Shikhman]{muller2021dynamic}
David M{\"u}ller, Yurii Nesterov, and Vladimir Shikhman.
\newblock Dynamic pricing under nested logit demand.
\newblock \emph{arXiv preprint arXiv:2101.04486}, 2021.

\bibitem[Okada and Tercieux(2012)]{okada2012log}
Daijiro Okada and Olivier Tercieux.
\newblock Log-linear dynamics and local potential.
\newblock \emph{Journal of Economic Theory}, 147\penalty0 (3):\penalty0
  1140--1164, 2012.

\bibitem[Palacios and Tetali(1996)]{PALACIOS1996119}
JoséLuis Palacios and Prasad Tetali.
\newblock A note on expected hitting times for birth and death chains.
\newblock \emph{Statistics \& Probability Letters}, 30\penalty0 (2):\penalty0
  119--125, 1996.
\newblock ISSN 0167-7152.
\newblock \doi{https://doi.org/10.1016/0167-7152(95)00209-X}.
\newblock URL
  \url{https://www.sciencedirect.com/science/article/pii/016771529500209X}.

\bibitem[Papadimitriou and Peng(2021)]{10.1145/3465456.3467616}
Christos Papadimitriou and Binghui Peng.
\newblock Public goods games in directed networks.
\newblock In \emph{Proc. of the 22nd ACM Conference on Economics and
  Computation (EC)}, page 745–762, 2021.
\newblock ISBN 9781450385541.
\newblock \doi{10.1145/3465456.3467616}.
\newblock URL \url{https://doi.org/10.1145/3465456.3467616}.

\bibitem[Penna(2018)]{DBLP:journals/ijgt/Penna18}
Paolo Penna.
\newblock The price of anarchy and stability in general noisy best-response
  dynamics.
\newblock \emph{Int. J. Game Theory}, 47\penalty0 (3):\penalty0 839--855, 2018.
\newblock \doi{10.1007/S00182-017-0601-Y}.
\newblock URL \url{https://doi.org/10.1007/s00182-017-0601-y}.

\bibitem[Sawa(2019)]{sawa2019stochastic}
Ryoji Sawa.
\newblock Stochastic stability under logit choice in coalitional bargaining
  problems.
\newblock \emph{Games and economic behavior}, 113:\penalty0 633--650, 2019.

\bibitem[Shao et~al.(2023)Shao, Cheung, and Huang]{10.1109/TNET.2023.3274114}
Qi~Shao, Man~Hon Cheung, and Jianwei Huang.
\newblock Crowdfunding with cognitive limitations.
\newblock \emph{IEEE/ACM Trans. Netw.}, 31\penalty0 (6):\penalty0 2714–2729,
  may 2023.
\newblock ISSN 1063-6692.
\newblock \doi{10.1109/TNET.2023.3274114}.
\newblock URL \url{https://doi.org/10.1109/TNET.2023.3274114}.

\bibitem[Soundy et~al.(2021)Soundy, Wang, Stevens, and Chan]{soundy2021game}
Jared Soundy, Chenhao Wang, Clay Stevens, and Hau Chan.
\newblock Game-theoretic analysis of effort allocation of contributors to
  public projects.
\newblock 2021.

\bibitem[van~de Geer and den Boer(2022)]{van2022price}
Ruben van~de Geer and Arnoud~V den Boer.
\newblock Price optimization under the finite-mixture logit model.
\newblock \emph{Management Science}, 68\penalty0 (10):\penalty0 7480--7496,
  2022.

\bibitem[Wang et~al.(2023)Wang, Wang, and Wei]{wang2023bundling}
Lu~Wang, Xue Wang, and Hang Wei.
\newblock Bundling and pricing strategies in crowdfunding.
\newblock \emph{Available at SSRN 4432999}, 2023.

\bibitem[Ward(2024)]{IOHK-PC-02}
Mike Ward.
\newblock Partner chains are coming to cardano.
\newblock \emph{IOHK blog post}, November 2024.
\newblock URL \url{https://iohk.io/en/blog/posts/2023/11/page-1/}.

\bibitem[Worldcoin(2025)]{worldcoin}
Worldcoin, 2025.
\newblock Description available at https://world.org/worldcoin-token.

\bibitem[Yaish and Livshits(2024)]{tierdrop}
Aviv Yaish and Benjamin Livshits.
\newblock Tierdrop: Harnessing airdrop farmers for user growth, 06 2024.

\bibitem[Yan and Chen(2021)]{yan2021optimal}
Xiang Yan and Yiling Chen.
\newblock Optimal crowdfunding design.
\newblock In \emph{Proc. of the 20th International Conference on Autonomous
  Agents and MultiAgent Systems (AAMAS)}, pages 1704--1706, 2021.

\bibitem[Yu et~al.(2020)Yu, Zhou, Brantingham, and
  Vorobeychik]{yu2020computing}
Sixie Yu, Kai Zhou, Jeffrey Brantingham, and Yevgeniy Vorobeychik.
\newblock Computing equilibria in binary networked public goods games.
\newblock In \emph{Proceedings of the AAAI Conference on Artificial
  Intelligence}, volume~34, pages 2310--2317, 2020.

\end{thebibliography}

\newpage
\appendix

\onecolumn

\section*{\centering \Huge  Appendix}

\section{Postponed Proofs}

\subsection{Proof of Theorem~\ref{th:potential-game}}
\begin{proof}
	Given an effort vector $a=(a_1,\ldots,a_n)$, let 
	us consider the  (potential) function in \eqref{eq:unifrom-rew-potential}.
	We show that this function is indeed an exact potential function. That is, for any $a$ and $a'=(a_i', a_{-i})$, we have
	\begin{align*}
		u_i(a) -  u_i(a') = & (r_i \cdot t(a) - c_i \cdot a_i) - (r_i \cdot t(a')- c_i \cdot a'_i) \\
	= & \gamma \cdot (t(a) - t(a') ) - c_i \cdot (a_i- a'_i) 
 \\
	= & \gamma \cdot (t(a) - t(a') ) - c_i \cdot (a_i- a'_i) + \sum_{j \neq i} c_j \cdot a_j - \sum_{j \neq i} c_j \cdot a_j
	\\ \intertext{and since $a_j'=a_j$ for all $j \neq i$, }
 = & \gamma \cdot (t(a) - t(a') ) - c_i \cdot (a_i- a'_i) + \sum_{j \neq i} c_j \cdot a_j' - \sum_{j \neq i} c_j \cdot a_j \\
	= & \gamma \cdot (t(a) - t(a') ) - SC(a) + SC(a') 
	\\ = & \phi(a) - \phi(a') 
	\end{align*}
which completes the proof. 
\end{proof}

\subsection{Proof of Theorem~\ref{th:Nash-equilibria}}
\begin{proof}
    The proof follows by the equilibrium condition in \eqref{eq:equilibrium-def} and by distinguishing between the two cases. Note that the monotonicity of $V(\cdot)$ implies that the difference $V(a) - V(a^+)$ can be negative for $a_i^+ > a_i$, which then leads to the first condition (upper bound on $\frac{\rho}{n}$) when applying \eqref{eq:equilibrium-def} with $a'=a^+$. 
\end{proof}

\subsection{Proof of Theorem~\ref{th:stable-equilibria}}
\begin{proof}
    From \cite{blume1993statistical,BLUME2003251}, the stationary distribution of the underlying Markov chain concentrates uniformly on the potential maximizers as $\beta \rightarrow \infty$. The result follows by considering the potential function in \eqref{eq:unifrom-rew-potential} together with the identity \eqref{eq:inv-monetary-rewards}. 
\end{proof}

\subsection{Proof of Theorem~\ref{th:mixing-birth-and-death}}
We first observe that, since we are considering \emph{symmetric} technology functions, the resulting dynamics boils down to a proper the birth and death  process. Recall that we keep track of the number $\ell$ of actively participating players in a given profile $a$, i.e., $\ell = \sum_i a_i$ and $a_i\in\{0,1\}$. The process is defined by the probabilities in \eqref{eq:birth-death-rate-general}, where $p_i^{\beta}(\cdot)$ is the logit response  \eqref{eq:logit-response}. Since we assume all players have the same cost and the same actions, and the technology function is symmetric, all these probabilities $p_i^{\beta}(\cdot)$ are the same. Hence the probabilities in \eqref{eq:birth-death-rate-general} are independent of $i$ and indeed yield a proper birth and death process.

Next observe that, using the  birth and death process in \eqref{eq:birth-death-rate-general}, we obtain a useful ``recursive'' formulation for the stationary distribution:

\begin{lemma}\label{le:stationary:recursive}
     For binary uniform costs, airdrop rewards, any symmetric technology function, and any inverse noise parameter $\beta>0$, the stationary distribution \eqref{eq:birth-and-death-stationary} of the birth and death process satisfies
     \begin{align}
         \hat{\pi}(\ell) = \hat{\pi}(0) \cdot \frac{\binom{n}{\ell}}{\exp(\alpha\beta \ell)} \cdot \exp\left(\beta \gamma(t(\ell)-t(0))\right) \label{eq:stationary:recursion} 
     \end{align}
     where $\hat{\pi}(\ell) = \binom{n}{\ell} \cdot \pi(\ell)$ is the stationary distribution of the corresponding birth and death process. 
\end{lemma}

\begin{proof}
Observe that $\hat{\pi}(0) = \frac{\exp(\beta\gamma t(0))}{Z}$ and, for every $\ell \in [0,n]$, we have 
   \begin{align*}
       \hat{\pi}(\ell) =  \frac{\binom{n}{\ell}}{Z} \cdot \frac{\exp(\beta\gamma t(\ell))}{\exp(\beta\alpha \ell)}  =  \frac{\binom{n}{\ell}}{Z} \cdot \frac{\exp(\beta\gamma t(0))}{\exp(\beta\alpha \ell)} \cdot \exp(\beta\gamma (t(\ell)-t(0)) \\ = \hat{\pi}(0) \cdot \frac{\binom{n}{\ell}}{\exp(\alpha\beta \ell)} \cdot \exp(\beta\gamma(t(\ell)-t(0)))
   \end{align*}
   which completes the proof.
\end{proof}
We are now in a position to prove Theorem~\ref{th:mixing-birth-and-death}
\begin{proof}[Proof of Theorem~\ref{th:mixing-birth-and-death}]
    Theorem~1.1 in \cite{chen2013mixing} states that
    \begin{align}
        \tmix(\epsilon) \leq \frac{18\tcut}{\epsilon^2} \ , && \tmix(1/10) \geq \frac{\tcut}{6} \ .
    \end{align}
    Since by definition  $\tmix=\tmix(1/4)$, the upper bound in \eqref{eq:tmix-general} follows immediately by taking $\epsilon=1/4$ in the first inequality above. As for the lower bound, we simply use  the well-known relation $\tmix(\epsilon) \leq \left\lceil \log_2 (1/\epsilon)\right \rceil \cdot \tmix $ (see e.g. \cite{LevinPeresWilmer2006}). Hence, $\tmix(1/10)\leq 4 \tmix$ which together with the second inequality above gives the lower bound in \eqref{eq:tmix-general}.
\end{proof}

\subsection{Proof of Theorem~\ref{thm:hitting-time-general}}
\begin{proof} 
    From Theorem~2.3 in \cite{PALACIOS1996119}, we have 
    \begin{align}\label{eq:hitting-PT96}
      \thit(\ell) \geq \thit(\ell_2+1) \geq    E_{\ell_2}T_{\ell_2+1}= \frac{1}{\hat{\pi}(\ell_2)p(\ell_2)}\sum_{\ell=0}^{\ell_2} \hat{\pi}(\ell)
    \end{align}
    where  $E_a T_b$ is the expected time to reach state $b$ for the first time starting from state $a$.  Since the drift in $I$ is at most $d_I$,
    \begin{align}
       \frac{\hat{\pi}(\ell_1)}{\hat{\pi}(\ell_2)}  = \frac{\hat{\pi}(\ell_1)}{\hat{\pi}(\ell_1+1)} \frac{\hat{\pi}(\ell_1+1)}{\hat{\pi}(\ell_1+2)} \cdots \frac{\hat{\pi}(\ell_2-1)}{\hat{\pi}(\ell_2)} \geq \left(\frac{1}{d_I}\right)^{\ell_2-\ell_1} \ .
    \end{align}
    This and \eqref{eq:hitting-PT96} imply the first part of the theorem. As for the second part, we show that the drift in this interval is at most $d_I = \exp(\beta(\frac{\rho}{n} \cdot s - \alpha))\cdot \frac{n-\ell_1}{\ell_1+1} $. We simply observe that 
    \begin{align}
        V(\ell_1+a+1) - V(\ell_1+a)  \leq s && 
    \end{align}
    for all $a \in [0,\ell_2-\ell_1-1]$. This and Lemma~\ref{le:stationary:recursive} imply
    \begin{align*}
    \frac{\hat{\pi}(\ell_1+a+1)}{\hat{\pi}(\ell_1+a)} = &
    \frac{\binom{n}{\ell_1+a+1}\cdot \exp(\beta(\gamma t(\ell_1+a+1) - \alpha \cdot (\ell_1+a+1)))}{\binom{n}{\ell_1+a}\cdot \exp(\beta(\gamma t(\ell_1+a) - \alpha \cdot (\ell_1+a)))} \\
    = & \frac{n-\ell_1+a}{\ell_1+a+1} \cdot \exp(\beta(\gamma(t(\ell_1+a+1)-t(\ell_1+a))-\alpha) 
    \\
    \stackrel{\eqref{eq:inv-monetary-rewards}}{=} & \frac{n-\ell_1+a}{\ell_1+a+1} \cdot \exp(\beta(\frac{\rho}{n}(V(\ell_1+a+1)-V(\ell_1+a))-\alpha) \\
    \leq & 
    \frac{n-\ell_1+a}{\ell_1+a+1} \cdot  \exp(\beta(\frac{\rho}{n} \cdot s - \alpha)) \leq \frac{n-\ell_1}{\ell_1+1} \cdot \exp(\beta(\frac{\rho}{n} \cdot s - \alpha))  \ .
    \end{align*}
      This completes the proof.
\end{proof}

\subsection{Proof of Theorem~\ref{th:pot-max-threshold}}
\begin{proof}
        We apply Theorem~\ref{th:stable-equilibria} and evaluate the set $\potmax_\rho$ of potential maximizers in our game depending on the rewards. We first observe that this set can only contain pure Nash equilibria, that is, the case $\ell=0$ or $\ell=\tau$, as every other state has a smaller potential.  The condition for $\ell=\tau$ being the highest potential is equivalent to 
    \begin{align}
        \phi(0) < \phi(\tau)
 && \stackrel{\eqref{eq:unifrom-rew-potential}}{\Longleftrightarrow} &&    \frac{\rho}{n}\low{V} < \frac{\rho}{n} \high{V} - \alpha \cdot \tau \end{align}
 which is the case in which the distribution selects the high value outcome ($PMAX_\rho$ contains only the profiles with $\ell=\tau$). The opposite inequality corresponds to the distribution  selecting the low value outcome ($\potmax_\rho$ contains only the profile with $\ell=0$). Finally, if equality holds, the distribution concentrates equally between both kind of pure Nash equilibria ($\potmax_\rho$ contains  the bad equilibrium with $\ell=0$ and the $\binom{n}{\tau}$ good equilibria in which $\ell=\tau$). 
\end{proof}

\subsection{Proof of Theorem~\ref{th:threshold-finite-beta:prob}}
\begin{proof}
We express the potential function \eqref{eq:unifrom-rew-potential} as a function of $\ell=\sum_i a_i$, 
\begin{align}\label{eq:pot-threshold}
         \phi(a) = 
         \begin{cases}\low{\phi}(\ell) & \ell < \tau \\
        \high{\phi}(\ell)  & \ell \geq   \tau
    \end{cases}  && \begin{array}{l}
\low{\phi}(\ell) = \gamma\low{t}  - \alpha \ell
\\ 
\high{\phi}(\ell) = \gamma\high{t}  - \alpha \ell
\end{array} \ . 
    \end{align}
    The stationary distribution is 
    \begin{align}\label{eq:stationary-threshold-state}
         \pi_\gamma(\ell) = \frac{1}{Z}\cdot
         \begin{cases}\exp(\beta\low{\phi}(\ell)) & \ell < \tau \\ 
        \exp(\beta\high{\phi}(\ell))  & \ell \geq   \tau
    \end{cases} 
    \end{align}
    where 
    \begin{align}
        Z = \sum_{\ell=0}^n \exp(\beta\phi(\ell)) \cdot  \binom{n}{\ell} = \underbrace{\sum_{\ell=0}^{\tau-1} \exp(\beta\low{\phi}(\ell)) \cdot  \binom{n}{\ell}}_{\low{Z}}
        + 
        \underbrace{\sum_{\ell=\tau}^n \exp(\beta\high{\phi}(\ell)) \cdot  \binom{n}{\ell}}_{\high{Z}} \\ 
        = \exp(\beta\gamma\low{t})\cdot \underbrace{\sum_{\ell=0}^{\tau-1} \exp(-\beta\alpha \ell)\cdot \binom{n}{\ell}  }_{\low{S}}
        + 
        \exp(\beta\gamma\high{t})\cdot \underbrace{\sum_{\ell=\tau}^n \exp(-\beta\alpha \ell) \cdot  \binom{n}{\ell}}_{\high{S}} \ . 
    \end{align}
    Finally observe that
    \begin{align}
        \high{p}(\rho) = \sum_{\ell = \tau}^n \pi_\gamma(\ell)  \cdot  \binom{n}{\ell} = \frac{\high{Z}}{\low{Z} + \high{Z}} = \frac{1}{1 + \low{Z}/\high{Z}}
    \end{align}
    and 
    \begin{align}
        \low{Z}/\high{Z} = \frac{\exp(\beta\gamma\low{t})\cdot \low{S}}{\exp(\beta\gamma\high{t})\cdot \high{S}}
    \end{align}
    which proves the first part for $C=\low{S}/\high{S}$. We also observe that $\high{p}(0) = \frac{1}{1 + \low{Z}/\high{Z}} = \frac{1}{1 + C}$, that is, $C=\frac{1}{\high{p}(0)}-1= \frac{1 - \high{p}(0)}{\high{p}(0)}$. This completes the proof.
\end{proof}

\subsection{Proof of Theorem~\ref{th:threshold:profit-logit}}
\begin{proof}
We first show the identity \eqref{eq:expected-profit-threshold} using Theorem~\ref{th:threshold-finite-beta:prob}:  \begin{align}
        \profit(\rho)  & =(1-\rho)\cdot \high{V}\cdot \high{p}(\rho) - d_V \stackrel{(Thm~\ref{th:threshold-finite-beta:prob})}{=} \high{V} \cdot \frac{1-\rho}{1 + C \cdot \exp(-\rho \cdot B)} - d_V\ . 
    \end{align}
Next consider  the derivative of the expression in \eqref{eq:expected-profit-threshold} with respect to $\rho$, 
    \begin{align}\label{eq:exp-profit-derivative}
        \frac{\partial}{\partial \rho}\left(
        \frac{1-\rho}{1+C\cdot \exp(-B\cdot \rho)}  \cdot \high{V}\right)
        = -\frac{\mathrm{e}^{B{\rho}} \left(\mathrm{e}^{B{\rho}} + C \left(B \left({\rho} - 1\right) + 1\right)\right)}{\left(\mathrm{e}^{B{\rho}} + C\right)^{2}} \ . 
    \end{align}
    The derivative above \eqref{eq:exp-profit-derivative} is negative for 
    \begin{align}
        B (\rho -1 ) + 1>0 && \Longleftrightarrow &&
        \rho > 1 - \frac{1}{B} = \bar{\rho}
    \end{align}
     and thus the optimal airdrop rewards $\rho^*$ must satisfy $\rho^* \leq \bar{\rho}$. 
    
    For $\rho =0$, the derivative above \eqref{eq:exp-profit-derivative} can be rewritten as 
    \begin{align}
    -\frac{1 + C \left(-B + 1\right)}{\left(1 + C\right)^{2}} = \frac{C \left(B - 1\right)-1}{\left(1 + C\right)^{2}} \ . 
    \end{align}
    This quantity is strictly positive for $B>1/C+1$, that is, $B=\beta\high{V}/n>1+\frac{\high{p}(0)}{1- \high{p}(0)} = \frac{1}{1- \high{p}(0)}$. In this case, by continuity, the derivative in \eqref{eq:exp-profit-derivative} is also strictly positive for all sufficiently small $\rho>0$. Hence, the optimal airdrop rewards $\rho^*$ must be strictly positive for $n < \beta \cdot \high{V} \cdot (1 - \high{p}(0))$. 
\end{proof}

\subsection{Proof of Theorem~\ref{th:mixing-LB-threshold}}

\begin{proof} We leverage on the bounds for birth and death chains (Theorem~\ref{th:mixing-birth-and-death}). To this end, we note that in this particular case of threshold functions, the birth and death probabilities \eqref{eq:birth-death-rate-general} are of the form 
\begin{align}\label{eq:birth-death-threshold}
     p(\ell) & := 
     \frac{n - \ell}{n} \cdot \begin{cases}
          p_{\alpha\beta} & \ell \neq \tau-1 \\
        p_{(\alpha-\frac{\rho}{n}\Delta_V)\beta} & \ell = \tau-1
     \end{cases} \\ 
       q(\ell) & := \frac{\ell}{n} \cdot \begin{cases}
         1- p_{\alpha\beta} & \ell \neq \tau \\
        1 - p_{(\alpha-\frac{\rho}{n}\Delta_V)\beta} & \ell = \tau
     \end{cases}   
\end{align}
where $\Delta_V = \high{V}-\low{V}$ and
\begin{align}
     p_{x}:=\frac{1}{1 + \exp(x)} \ .
\end{align}

We show below that, because $\high{p}(\rho)>1/2$, there exists  $\ell_0\geq \tau$ satisfying the condition of Theorem~\ref{th:mixing-birth-and-death}. 
This  yields the desired lower bound as follows:
\begin{align}
    \tcut \stackrel{(Thm~\ref{th:mixing-birth-and-death})}{\geq} \sum_{\ell = 0}^{\tau-1} \frac{\hat{\pi}([0,\ell])}{\hat{\pi}(\ell)p(\ell)} \geq 
    \frac{\hat{\pi}(0)}{\hat{\pi}(\tau-1)}  \stackrel{(Lem~\ref{le:stationary:recursive})}{=} \frac{\exp(\alpha\beta(\tau-1))}{\binom{n}{\tau-1}} \ . 
    \end{align}
It thus only remains to show the claim about $\ell_0$. We simply observe that, by definition of threshold technology \eqref{eq:threshold-t}, we have
\begin{align}
    \hat{\pi}([0,\tau-1]) =  1 - \high{p}(\rho) < 1/2
\end{align}
and therefore we can find $\ell_0 \geq \tau+1$ such that  $\hat{\pi}([0,\tau]) \geq 1/2$ as required by Theorem~\ref{th:mixing-birth-and-death}. This completes the proof. 
\end{proof}

\subsection{Proof of Theorem~\ref{th:hitting-threshold-UB-LB} (Part~\ref{th:hitting-threshold-UB-itm})}
\begin{proof}
The proof goes through the following steps:
\begin{enumerate}
       \item Theorem~2.3 in \cite{PALACIOS1996119} states that 
    \begin{align}\label{eq:hitting-progress}
      E_{\ell}T_{\ell+1}= \frac{1}{\hat{\pi}(\ell)p(\ell)}\sum_{h=0}^{\ell} \hat{\pi}(h)
    \end{align}
    where  $E_a T_b$ is the expected time to reach state $b$ for the first time starting from state $a$.
    \item Definition of $\ell^*$ in \eqref{eq:equilibrium-alphabeta} implies that the drift (Definition~\ref{def:drift}) in the interval $I = [0, \ell^*]$ is at least $d_I = 1$, that is, 
    $\frac{\hat{\pi}(\ell+1)}{\hat{\pi}(\ell)}\geq 1$ for all $\ell \in [0,\ell^*-1]$. Therefore, 
    \begin{align}
        \hat{\pi}(0) \leq
        \hat{\pi}(1) \leq \cdots \leq \hat{\pi}(\ell^*)
    \end{align}
which allows us to bound the quantities in \eqref{eq:hitting-progress} as follows:  
for all $\ell \in [0,\ell^*-1]$ we have
\begin{align}
    E_hT_{\ell+1} \leq (\ell+1) / p(\ell+1) \stackrel{\eqref{eq:birth-death-threshold}}{\leq}   (\ell + 1)\frac{n}{n-\ell} (1 + \exp(\alpha\beta)) 
    \\ \leq \ell^*\frac{n}{n-\ell^*} (1 + \exp(\alpha\beta)) \stackrel{\eqref{eq:equilibrium-alphabeta}}{= } \frac{n^2}{n-\ell^*}  \ . 
\end{align}
    \item Putting things together, and by definition $E_aT_b$, we have 
    \begin{align}
        \thit(\ell^*) = E_0T_{\ell^*}    = E_0T_1 + E_1T_2 + \cdots + E_{\ell^*-1}T_{\ell^*}\leq (\ell^*+1) \frac{n^2}{n-\ell^*} 
    \end{align}
    where in the second identity we use linearity of expectation. 
\end{enumerate}
 
\end{proof}

\subsection{Proof of Theorem~\ref{th:hitting-threshold-UB-LB} (Part~\ref{th:hitting-threshold-LB-itm})}
\begin{proof}
Since the threshold technology  is $0$-steep (``flat''),  Theorem~\ref{thm:hitting-time-general} implies the first lower bound in Part~\ref{th:hitting-threshold-LB-itm}:
\begin{align}\label{eq:hitting-threshold:lb-ell}
        \thit(\tau) \geq \left(\exp(\alpha\beta)\cdot \frac{\ell+1}{n-\ell}\right)^{\tau-\ell } && \text{for all } 0 \leq \ell \leq \tau \ .  
    \end{align} We next prove the second lower bound in Part~\ref{th:hitting-threshold-LB-itm}: \begin{align}\label{eq:hitting-threshold:lb}
        \thit(\tau) \geq (1+1/\ell^*)^{\tau-\ell^* -1 } \ . 
    \end{align} We take $\ell^* \leq \ell \leq \ell^* + 1$ and apply the lower bound in \eqref{eq:hitting-threshold:lb-ell}. Observe that the definition of $\ell^*=\frac{n}{1 + \exp(\alpha\beta)}$ in \eqref{eq:equilibrium-alphabeta} implies $n - \ell^* = \frac{n\cdot \exp(\alpha\beta)}{1 + \exp(\alpha\beta)}$ and thus, since $\ell \geq \ell^*$, we get 
\begin{align}
        \exp(\alpha\beta)\cdot \frac{\ell+1}{n-\ell}\geq \exp(\alpha\beta)\cdot\frac{\ell^*+1}{n-\ell^*} = \exp(\alpha\beta)\cdot\frac{\frac{n}{1 + \exp(\alpha\beta)}+1}{\frac{n\cdot \exp(\alpha\beta)}{1 + \exp(\alpha\beta)}} \\ = 1 + \frac{1 + \exp(\alpha\beta)}{n} = 1 + \frac{1}{\ell^*} \ . 
    \end{align} 
    By plugging this into  \eqref{eq:hitting-threshold:lb-ell}, and using $\ell \leq \ell^*+1$,  we get  the desired lower bound \eqref{eq:hitting-threshold:lb}.
\end{proof}

\newpage

\section{Further Applications: governance, partnerchains  or \\ new blockchains}\label{sec:partnerchain}

Anecdotal evidence from the field suggests that different projects have different costs and different success probabilities. For example, launching a new chain building on an existing blockchain is more economical and likelier to succeed, than a brand new blockchain. 

At the lowest end of the cost spectrum, governance projects tend to be cheap, corresponding to a low $\alpha$ in our model. Note that a low $\alpha$ makes $\ell^*$ higher \eqref{eq:equilibrium-alphabeta}, which makes convergence to $\ell^*$ faster. Revisiting Figure~\ref{fig:hitting-revisited}, a low $\alpha$ is likely to place the system in the fast hitting region (the interval from 0 to approximately $0.5$ in the graph) where the hitting time is growing approximately linearly, with a low slope, in the cost.
The implication is that it is easier for a system designer to make such projects work. Conventional wisdom in the field is in line with this result, as governance projects are perceived to be easier to succeed.

\begin{figure}
    \centering
    \includegraphics[width=0.6\linewidth]{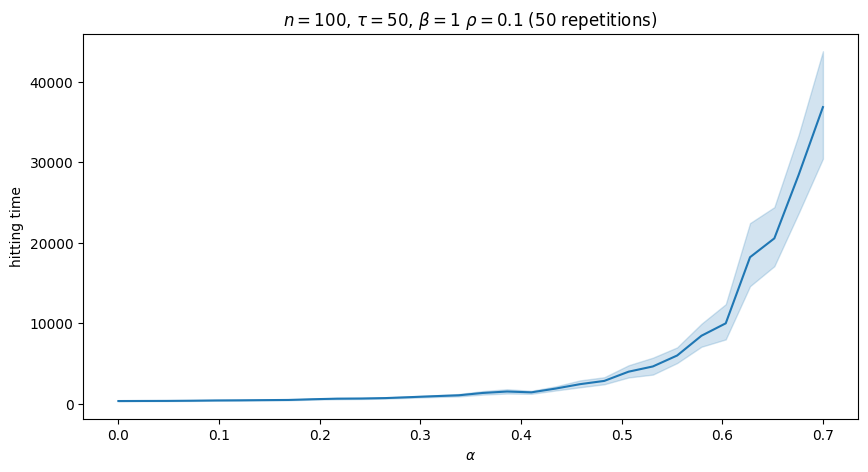}
    \caption{Hitting time for threshold technologies with increasing costs $\alpha$ (same as Fig.~\ref{fig:dynamics-time}(left) in main text reproduced here for convenience).}
    \label{fig:hitting-revisited}
\end{figure}

In the middle of the spectrum we have layer 2 projects, such as restaking. That is the region in the figure between 0.5 and 0.6, where there is a sharper rise in the hitting time.
A common example of such projects would be \emph{sidechains}, which rely on an existing mainchain to reduce costs and complexity of operations, as well as inherit some security guarantees directly from the mainchain. Note that in this setting, the mainchain plays no ``active role''. Some recent designs suggest having \emph{partnerchains} of a mainchain where the latter provides some ``help'' to the newly birthed chain during its launch phase. This can be in the form of resources (e.g., validators) from the mainchain that are partially contributing also to the partnerchain. One of the benefits of these designs is the \emph{reduced costs} for the contributors (validators) who can re-use their skills and tools for the partnerchain (as opposed to the case in which these skills and tools have to be acquired from scratch). A concrete example related to Cardano can be found in Section~\ref{sec:cardano}.

At the other end of the spectrum, layer 1 projects such as brand new blockchains will have a substantially higher cost of contribution $\alpha$, as the contributors will need new knowledge, new software and possibly even new hardware to be able to contribute in the first place. These projects are notoriously hard to launch in the field, which is in line with our model.

\subsection{Example: Cardano partnerchain framework}\label{sec:cardano}

The next  example considers Cardano as the mainchain and its ``validators'', i.e.  \emph{stake pool operators (SPOs)}, to as potential contributors who will be validating blocks for the new chain.

\begin{example}[Cardano partnerchain framework]\label{ex:Cardano-Midnight}
   Consider Cardano as the mainchain and a partnerchain inviting  Cardano's SPOs to validate its blocks during its launch period:
\begin{itemize}
    \item SPOs stake mainchain tokens, and receive some amount of partnerchain tokens for running some additional (partnerchain) software to validate partnerchains blocks;
    \item SPOs cannot be staked as the mainchain cannot verify activities inside the partnerchain; the latter can only distribute rewards, but cannot (directly) punish misbehaving validators;
    \item The main utility of  SPOs is in making the partnerchain successful, so that the price of its tokens (rewards) reaches a desired value after a target period of time.
\end{itemize}
\end{example}
The first key feature of the above parnerchain framework, is the \emph{reduced cost} for contributing to the new chain, specifically, SPOs will encounter only a moderate additional cost, as they already possess the required hardware and knowledge to contribute to the partnerchain:
\begin{quote}
    \emph{[...] SPO participation: any Cardano SPO can become a validator for a partner chain with minimal costs related to hardware upgrades, no required permissions, and no software expenses. The toolkit simplifies this configuration process. \citep{IOHK-PC-02}}  
\end{quote}
The fact that their activity and performance for the partnerchain does not put any risk on the staked tokens in the mainchain, is an additional factor that contributes to reduce the costs (measured as risk of being staked). 
In this work we investigate the following questions:
\begin{itemize}
    \item \emph{Do higher rewards suffice to achieve enough collaboration and thus  successful launch?}
    \item \emph{Is there an actual benefit in the partnerchain design with reduced costs?}
\end{itemize}
\begin{remark}[non-verifiable efforts]\label{rem:non-verifiable}
    We shall consider the case in which, for security and technical reasons, the mainchain is mainly responsible during the launch phase of the sidechain. In particular,  rewards are provided by the  mainchain  and the latter cannot monitor the efforts $a_i$ of the players. 
\end{remark}

In this section, we dive into two key parameters that affect the dynamics, namely, the rewards $\gamma$ and the costs $\alpha$. Though rewards can always be increased, we show below that their effect is limited to the dynamics remaining in the good region (above the threshold $\tau$). In particular, the \emph{time} to reach that region depends on the costs $\alpha$, the smaller the better,  and it is \emph{independent} of rewards $\gamma$. A key feature of \emph{partnerchains} designs is the ability to \emph{reduce} these costs $\alpha$ in two ways (and combinations thereof):
\begin{itemize}
    \item By involving existing validators of the mainchain in the partnerchain;  
    \item By providing \emph{multi-token} payments, meaning that also some amount of \emph{mainchain tokens} are rewarded. 
\end{itemize}

\

\subsubsection{Multi-token rewards to lower the costs} \label{sec:multi-token-rewards}
We consider the extension in which rewards consist of both the new native token and another currency with a relatively stable price (the latter introduces a cost for the designer but effectively reduces $\alpha$). Within the Cardano partnerchain framework (Example~\ref{ex:Cardano-Midnight}), multi token rewards have been proposed:
\begin{quote}
\emph{
    [...] Babel fees will solve tokenomics issues for new networks while allowing SPOs to be compensated in ada. \citep{IOHK-PC}}
\end{quote}

\section{Other technology functions}\label{sec:other-technoligies}
In this section, we discuss a number of natural technology functions, and extend the analysis in Section~\ref{sec:threshold} to some of them.  

\subsection{Quadratic technologies (Metcalfe's Law)}
We consider the symmetric technology following Metcalfe’s Law \citep{ALABI201723}, that is, a quadratic growth:
\begin{align}
\label{eq:quadratic-technology}
V(\ell) = \frac{\ell^2}{\tau}
\end{align}
where the rescaling parameter  $\tau>0$ controls the region where the system has a ``slow start'' because for small $\ell$ one additional contribution results in a ``moderate'' increase in the system value: $ V(\ell+1) - V(\ell) = \frac{2\ell +1}{\tau}$. This is analogous to the threshold function \eqref{eq:threshold-t}, where additional contributions provide no system increase at all for $\ell <\tau$ . 

\begin{theorem}
    For any Metcalfe's Law (quadratic) technology \eqref{eq:quadratic-technology} with airdrop rewards \eqref{eq:rewards} the following holds:
    \begin{enumerate}
        \item For $\alpha < \frac{1}{\tau n}$, there is a bad pure Nash equilibrium for $\rho \leq \alpha \tau n$, while for $\rho >  \alpha \tau n$ there is  only the  good pure Nash equilibrium. 
        \item For $ \frac{1}{\tau n} < \alpha < \frac{1}{\tau}$, there exist both a bad and a good pure Nash equilibrium for all airdrop rewards $\rho$. Moreover, for $\rho> \alpha \tau$, logit dynamics for vanishing noise ($\beta \rightarrow \infty$) select only the good equilibrium.   
        \item For $\alpha > \frac{1}{\tau}$, there is a bad pure Nash equilibrium, and logit dynamics for vanishing noise ($\beta \rightarrow \infty$) select this bad equilibrium no matter the airdrop rewards $\rho$. 
    \end{enumerate}
\end{theorem}
\begin{proof}
    We observe the following
    \begin{enumerate}
        \item The bad pure Nash equilibrium exists if and only if $\rho \leq \alpha n \tau$ for $\rho \in [0,1]$. 
        \item The good pure Nash equilibrium exists if and only if $\rho \geq \frac{\alpha n \tau}{2m-1}$ for $\rho \in [0,1]$. 
    \end{enumerate}
    As for the logit dynamics, we evaluate the set of potential maximizers in Theorem~\ref{th:stable-equilibria}, which for the technology in \eqref{eq:quadratic-technology} under consideration is
    \begin{align}\label{eq:potential-max-quadratic}
     \potmax_\rho =  \argmax_\ell \left\{\frac{\rho}{n} \frac{\ell^2}{\tau}- \alpha \ell \right\}   
     \end{align}
    and we observe that for $\alpha > \frac{1}{\tau}$ the expression in \eqref{eq:potential-max-quadratic} is negative for all $\ell > 0$, and thus it is maximized for $\ell = 0$, the bad pure Nash equilibrium. Conversely, for $\alpha < \frac{1}{\tau}$ we can distinguish two cases depending on $\rho$: For $1 > \rho > \alpha \tau$, the expression in \eqref{eq:potential-max-quadratic} is larger than  $\alpha(\frac{\ell^2}{n} - \ell)$ and thus it is maximized for $\ell = n$, the good pure Nash equilibrium. For $\rho < \alpha \tau$, the expression in \eqref{eq:potential-max-quadratic} is less than $\alpha(\frac{\ell^2}{n} - \ell)$ and thus it is maximized for $\ell = 0$. 
\end{proof}

\subsection{Linear technologies}
	Consider the linear technology function of the form $V(a) = \lambda_V \cdot \ell$, where $\ell = \sum_i a_i$ and $a_i \in \{0,1\}$. Suppose further that $c_i = \alpha$ for all $i$. Then the largest system achievable value at equilibrium for an airdrop reward $\rho$ is 
	\begin{align}
		V^*(\rho) = 
		\begin{cases}
			\lambda_V \cdot n & \text{if } \rho > \frac{n\cdot \alpha}{\lambda_V} \\
			0 & \text{if } \rho < \frac{n\cdot \alpha}{\lambda_V}  
		\end{cases}\ .
	\end{align}
	Linear technologies with uniform costs exhibit an ``all-or-nothing'' property: we can incentivize all or none to contribute, depending on the ratio  $\frac{n \cdot \alpha}{\lambda_V}$ and the reward $\rho\in[0,1]$. In particular, for $\frac{n \cdot \alpha}{\lambda_V}>1$ the project will ``fail'' as it is impossible to incetivize players, while for  $\frac{n \cdot \alpha}{\lambda_V}\leq 1$ the designer can optimize her revenue by setting  $\rho$ equal to this ratio. This gives her a revenue equal to 
	$(1-\rho)\cdot V^* - d_V = (1-\frac{n \cdot \alpha}{\lambda_V})\cdot \lambda_V \cdot n -d_V = \lambda_V \cdot n - \alpha \cdot n^2 - d_V$, while each  player gets a monetary reward $\alpha \cdot n$.

We next consider logit dynamics and  its equilibria for this technology. 
The corresponding potential function \eqref{eq:unifrom-rew-potential} is 
\begin{align}
         \phi(\ell) = \frac{\rho}{n} \lambda_V \ell - \alpha \ell = \left(\frac{\rho}{n}\lambda_V - \alpha\right) \ell \ . 
    \end{align}
    For $\Gamma:= \frac{\rho}{n}\lambda_V - \alpha$, the stationary distribution of the corresponding birth and death process is  
    \begin{align}
       \hat{\pi}(\ell) = \binom{n}{\ell} \cdot \pi_\gamma(\ell) =  \binom{n}{\ell} \cdot \frac{\exp(\beta\phi(\ell))}{Z}=\binom{n}{\ell} \cdot\frac{\exp(\beta\ell \Gamma)}{Z}
    \end{align}
    where 
    \begin{align}
        Z = \sum_{\ell=0}^n  \binom{n}{\ell} \cdot  \exp(\beta\phi(\ell))  =  \sum_{\ell=0}^n \binom{n}{\ell} \cdot   \exp(\beta\ell\Gamma) = (1 + \exp(\beta \Gamma))^n
    \end{align}
    where the last identity is due to the Binomial Theorem. 
    Note that  for any airdrop reward $\rho > \frac{\alpha n}{\lambda_V}$ we have $\Gamma>0$ which favors states with higher $\ell$. Also remember that the designer aims at maximizing the expected profit  according to  \eqref{eq:profit-expected}. 
    
\begin{theorem}\label{th:hitting-linear}
    For any linear technology, the expected time $\thit(\ell)$ for the dynamics to reach the  state where $\ell$ players contribute can be lower bounded as
    \begin{align}
        \thit(\ell) \geq \left(\exp(-\Gamma\beta)\cdot \frac{\ell_1+1}{n-\ell_1}\right)^{\ell-\ell_1-1}\ , && \Gamma = \frac{\rho}{n}\lambda_V - \alpha \ . 
    \end{align}
\end{theorem}

\begin{proof}
We apply Theorem~\ref{thm:hitting-time-general} and observe that the function is $\lambda_V$-steep. For $I =[\ell_1,\ell_2]$ with $\ell_2=\ell-1$, we have $d_I=\exp(\beta(\frac{\rho}{n}\lambda_V - \alpha))\cdot \frac{n-\ell_1}{\ell_1+1}$, and the lower bound on the hitting time follows from the first part of  Theorem~\ref{thm:hitting-time-general}.
     
\end{proof}

Here the intuition is analogous to the threshold technology: The dynamics converge quickly to some average level $\ell^* = \frac{n}{p_{-\Gamma \beta}}= \frac{n}{\exp(-\Gamma \beta)}$ and from there it takes exponential time to reach higher values. In particular, increasing the steepness $\lambda_V$ increases $\ell^*$ (see Figure~\ref{fig:linear-lam-reps}). 

\begin{figure}
    \centering
    \includegraphics[scale=.5]{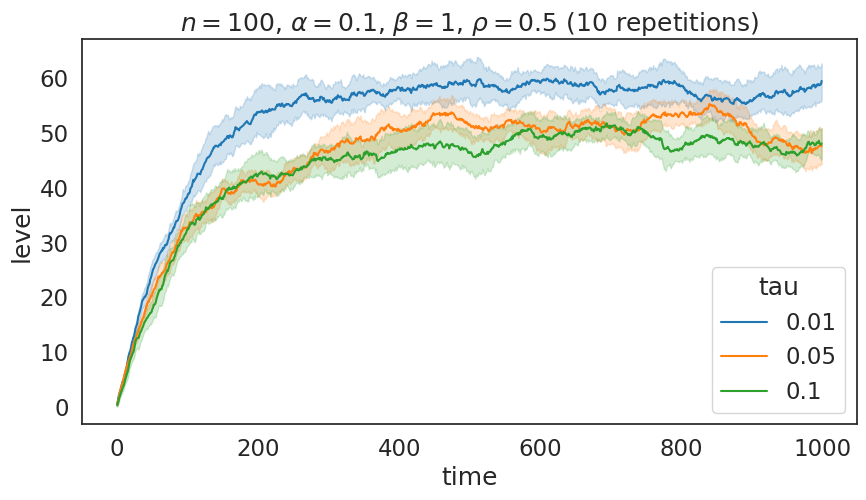}
    \caption{Linear technologies for different $\lambda_V\in\{10,20,100\}$ ($\lambda_V = 1/\tau$, 10 repetitions and 95\% accuracy).}
    \label{fig:linear-lam-reps}
\end{figure}

\subsection{S-shaped technologies}
S-shaped technology functions  stemming from P2P systems have been considered in  \cite{buragohain2003game} who studied the pure Nash equilibria in the resulting game. Here we consider a slightly different S-shaped technology functions, inspired by the previous work, namely:
\begin{align}
    V(\ell) = \frac{(\ell/\tau)^c}{1 + (\ell/\tau)^c}
\end{align}
for $c>0$ and $\ell = \sum_i a_i$ with $a_i \in [0,1]$. These functions have been used to model the probability of users in a P2P system of having their request accepted depending on the contribution level, though in \cite{buragohain2003game} this depends on the player strategy $a_i$ only.

\subsection{Concave technologies}
These technologies have been considered in \cite{bramoulle2007public}. That work considers the value being encoded by some underlying social network (graph). For binary efforts and a social network being fully connected (complete)  graph, we obtain a system value $V(\ell)$ where $V(\cdot)$ is concave. An example of such technologies is $V(\ell) = (1/\tau)\cdot \ell^c$ with $ c< 1$. We note that such technologies naturally exhibit a ``saturation'' phenomenon, since for large $\ell$ the marginal contribution of adding an extra effort does not compensate the required cost $\alpha$.

    


\section{Discussion and Further Material}\label{sec:further}

\subsection{Example of profit optimal airdrop with logit}
 Consider the AND technology with two players in Example~\ref{ex:no-bad-eq}, that is, the threshold technology in the special case $\tau=n$. In this case, we can express all quantities as a function of the number $\ell$ of actively participating players. In particular, 
 \begin{align}
   V(\pi_\rho) = \sum_{\ell=0}^n \pi_\rho(\ell)\cdot V(\ell) \cdot  \binom{n}{\ell} && \pi_\rho(\ell) = & \frac{\exp(\beta\phi(\ell))}{Z} 
 \end{align}
 where $\phi(\ell)=\phi(a)$ for $\ell = \sum_i a_i$, and $Z$ is the normalizing constant, i.e., $Z = \sum_a \exp(\beta\phi(a))$. For the case of $n=2$ players, and assuming $V(0)=V(1) = 0$, we get this simpler expression:
 \begin{align}
   V(\pi_\rho) & = \high{V}\cdot \frac{ 
   \exp(\beta \rho \high{V}/2)}{\exp(2\alpha\beta) + 2 \cdot \exp(\alpha\beta) + 
   \exp(\beta \rho \high{V}/2)} 
 \end{align}
 where $\high{V}:=V(2)$. 

\subsection{Logit dynamics and plausible parameter values}\label{sec:logit-values}

The traditional bibliography in decision and game theory assumes that between two choices, people will always choose the one with the highest expected payoff. This has been questioned when real people are involved in decision making, in the lab or in the field.

The seminal work of \cite{goeree2001ten} shows that in lab experiments, small payoff changes can make a very large difference in play in games, even when game theory predicts no change. Consider for example a matching pennies game  in Figure~\ref{fig:matchingpennies}  -- same as Table~1 in \cite{goeree2001ten}.


\begin{figure}
	\begin{center}
		\begin{matrixgame}[top={Left (48\%), Right (52\%)}, left={Top (48\%), Bottom (52\%)}, bottom={Symmetric Matching Pennies}]
			\payoffmatrix[w=2.5]{80,40 & 40, 80 \\ 
				40, 80 & 80, 40}
		\end{matrixgame}

		\begin{matrixgame}[top={Left (16\%), Right (84\%)}, left={Top (96\%), Bottom (4\%)}, bottom={Asymmetric Matching Pennies}]
			\payoffmatrix[w=2.5]{320,40 & 40, 80 \\ 
				40, 80 & 80, 40}
		\end{matrixgame}
		
		\begin{matrixgame}[top={Left (80\%), Right (20\%)}, left={Top (8\%), Bottom (92\%)}, bottom={Reversed Asymmetry}]
			\payoffmatrix[w=2.5]{44,40 & 40, 80 \\ 
				40, 80 & 80, 40}
		\end{matrixgame}
	\end{center}
\caption{Three different matching pennies games (with choices percentages) as in \cite{goeree2001ten}.}
\label{fig:matchingpennies}
\end{figure}

In the standard version, with symmetric payoffs, play follows the Nash equilibrium, with 48\% choosing left and 52\% choosing right. Changing the payoff for player one in UL from 80 to 320, changes behavior drastically for both players. 96\% choose top and 84\% choose right. Interestingly, the mixed strategy Nash equilibrium prediction changes for player 2 but not for player 1 (who plays to keep 2 indifferent). A game-theoretically unimportant change in payoffs has brought actual behaviour from very close to Nash to something not even resembling Nash.

Based on these facts, recent literature proposes concepts that evaluate the payoffs of all choices, and predict play probabilities that are related to these payoffs. A particularly successful specification is logit. In games, logit can be used as part of the Quantal Response Equilibrium concept, where players better respond (instead of best respond) given their beliefs, using choice probabilities according to logit. This is an equilibrium concept, since players perfectly anticipate each other's choice probabilities.

The application of logit is not limited to QRE. \cite{crawford2007level} show that logit can improve the fit of non-equilibrium models such as level-k in auctions. \cite{georganas2011english} finds that models with logit errors generally outperform models with normal errors, including when fitting the Nash equilibrium model itself.

Logit models have been criticized for being too flexible, in the sense that free parameter choice can fit many types of behavior. For this application, we restrict the parameter values to those found in earlier work. In particular, \cite{georganas2011english}  fits a series of models with logit errors. The parameter values that gave the best fit in the heterogeneous level-k model were 1.1 and 1.13. The fact that different types (level 1 and level 2) with very different predicted bidding, were found to have a very similar logit parameter (1.13) is encouraging.

\subsection{Further experiments and graphs}
In this section we provide additional experiments to convey further intuition about the convergence time and the distribution of the logit dynamics for threshold functions.

\begin{figure}
    \centering
    \includegraphics[scale=.5]{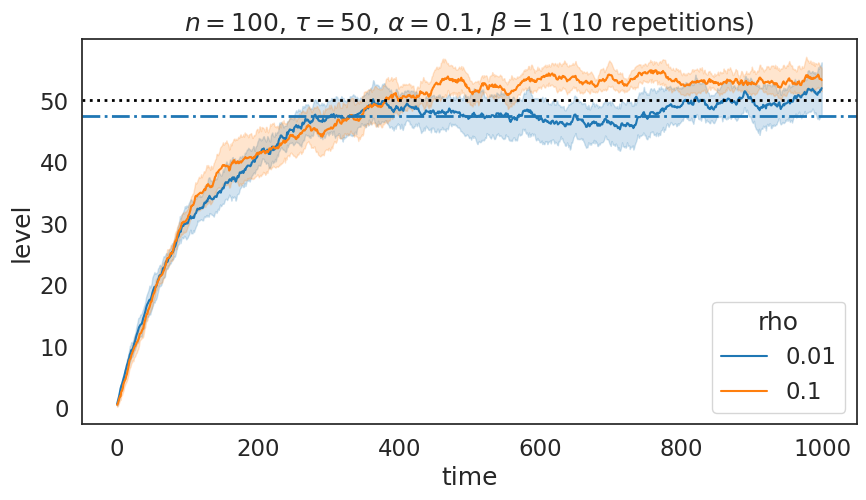}
    \caption{Higher airdrop allocations $\rho$ help stabilizing the dynamics in the ``high value'' region.}
    \label{fig:threshold-rho-reps}
\end{figure}

\begin{figure}
    \centering
    \includegraphics[scale=.5]{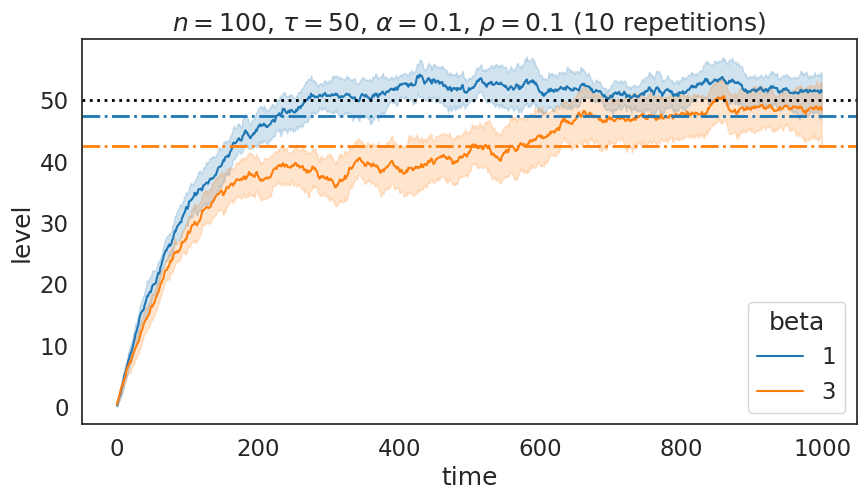}
    \caption{Higher $\alpha\beta$ slows down the convergence to the ``high value'' region (here $\alpha =0.1$ and $\beta \in \{1,3\}$).}
    \label{fig:dynamics-different-alphabeta}
\end{figure}

\begin{figure}[h!]
    \centering
    \includegraphics[scale=0.4]{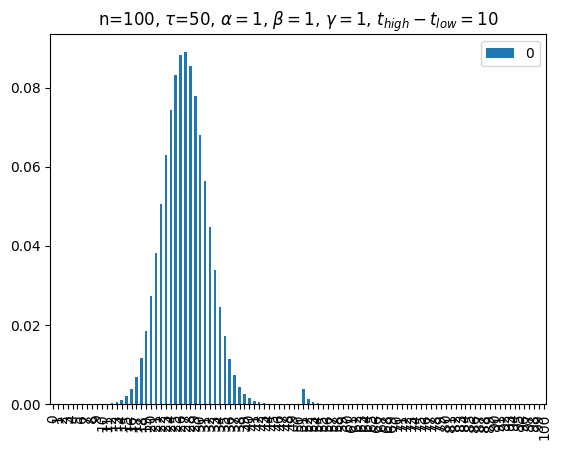}
    \includegraphics[scale=0.4]{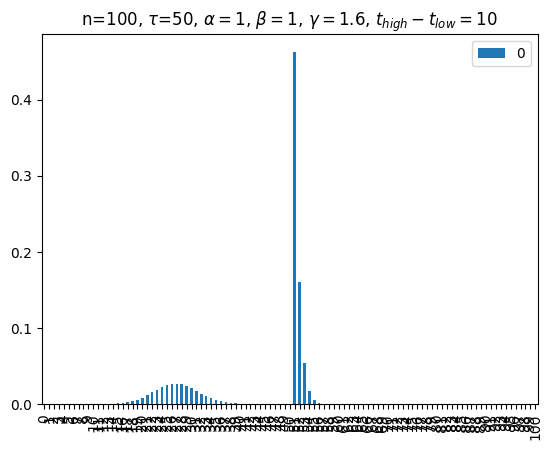}

    \includegraphics[scale=0.4]{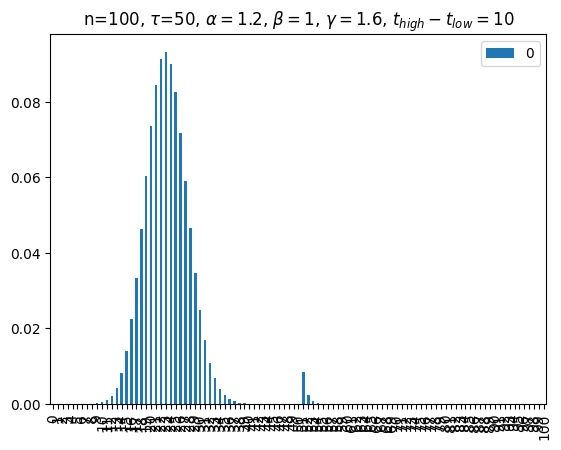}
\includegraphics[scale=0.4]{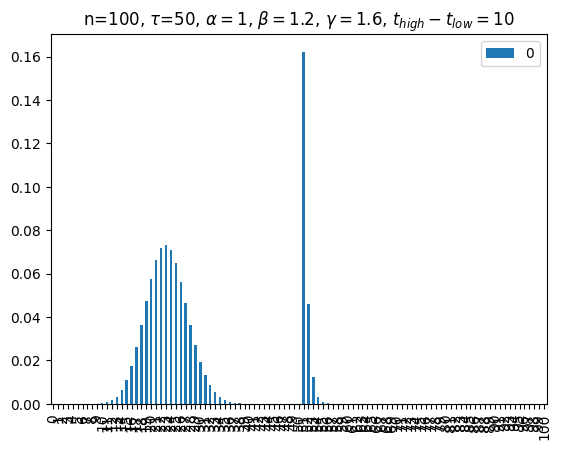}
    \caption{The stationary distribution of the birth and death process for some parameter combinations:  Larger rewards $\gamma$ increase the probability of success (top). Larger costs or larger inverse noise decreases the probability of success (bottom). The latter phenomenon is due to the rewards being set below some critical threshold.}
    \label{fig:bd-process-n=60}
\end{figure}

\end{document}